\documentclass{article}
\usepackage{graphicx} 
\usepackage{amsmath, amssymb, mathtools}
\usepackage{amsthm}
\usepackage[shortlabels]{enumitem}
\usepackage{xcolor}
\usepackage{subcaption}
\usepackage[a4paper, margin=1in]{geometry}
\usepackage[T1]{fontenc}
\usepackage[utf8]{inputenc}
\usepackage{authblk}
\usepackage{natbib, hyperref}
\usepackage{booktabs}
\usepackage{verbatim}
\usepackage{comment}
\usepackage{multirow}
\usepackage{algorithm}
\usepackage{algpseudocode}
\newtheorem{definition}{Definition}
\newtheorem{assumption}{Assumption}
\newtheorem{theorem}{Theorem}
\newtheorem{proposition}{Proposition}
\newtheorem{lemma}{Lemma}

\title{Monopoly Pricing of Weather Index Insurance}
\author{Tim J. Boonen \thanks{Department of Statistics and Actuarial Science, School of Computing and Data Science, The University of Hong Kong, Pokfulam Road, Hong Kong, P.R. China. Email: \url{tjboonen@hku.hk}.} \quad Wenyuan Li  \thanks{Department of Statistics and Actuarial Science, School of Computing and Data Science, The University of Hong Kong, Pokfulam Road, Hong Kong, P.R. China. Email: \url{wylsaas@hku.hk}.} \quad Zixiao Quan \thanks{Corresponding author. Department of Statistics and Actuarial Science, School of Computing and Data Science, The University of Hong Kong, Pokfulam Road, Hong Kong, P.R. China. Email: \url{zixiao.quan@connect.hku.hk}.}} 

\begin{document}

\maketitle
\noindent  
\begin{abstract}
    This study models the monopoly pricing of weather index insurance as a Bowley-type sequential game involving a profit-maximizing insurer (leader) and a 
    farmer (follower). 
    The farmer chooses an insurance payoff to minimize a convex distortion risk measure, 
    while the insurer anticipates this best response and selects a premium principle and its parameters to maximize profit net of administrative costs. For the insurer, we adopt three different premium-principle parameterizations: (i) an expected premium with a single risk-loading factor, (ii) a two-parameter distortion premium based on a power transform, and (iii) a fully flexible pricing kernel drawn from the general Choquet integral representation with nondecreasing distortions. For the farmer, we model index payoffs using neural networks and compare solutions under fully connected architectures with those under convolutional neural networks (CNNs). 
    We solve the game using a penalized bilevel programming algorithm that employs a function-value-gap penalty and delivers convergence guarantees without requiring the lower-level objective to be strongly convex.  Based on Iowa's soybean yields and high-dimensional PRISM weather data, we find 
    that CNN-based designs yield smoother, less noisy payoffs that reduce basis risk and push insurer profits closer to indemnity insurance levels. 
    Moreover, expanding pricing flexibility from a single loading to a two-parameter distortion premium, and ultimately to a flexible pricing kernel, systematically increases equilibrium profits. 
\end{abstract}

\section{Introduction}
Weather-based risks are one of the most significant challenges faced by agricultural producers worldwide, particularly in regions vulnerable to climate variability and extreme weather events. These risks can lead to substantial financial losses, which threaten the livelihoods of farmers and harm the socioeconomic stability of rural communities. Traditional risk management tools, such as crop indemnity insurance, have been widely adopted to mitigate these financial impacts. However, traditional insurance products often suffer from limitations, which include high administrative costs, moral hazard, and basis risk. To address these challenges, index insurance has emerged as a promising alternative. Index insurance connects insurance payoffs to objective and verifiable weather indices, which reduces the reliance on subjective loss assessments and mitigates moral hazards. Despite its potential, designing index contracts that materially hedge farmer risk while remaining profitable is especially challenging when the relationships between weather indices and crop yield are high-dimensional and nonlinear \citep{schlenker2009nonlinear, rigden2020combined}. These challenges are amplified in concentrated markets in which a monopolistic insurer can exert pricing power over contract terms and premium principles.

This paper studies the monopoly pricing of weather index insurance as a sequential game between a profit-maximizing insurer and a risk-averse farmer. We formulate the interaction as a bilevel optimization problem: the farmer faces a stochastic production loss and chooses an insurance payoff to minimize its distortion risk measure; the insurer anticipates this best response and selects a premium principle (and its parameters) to maximize profit net of administrative costs. We begin with two canonical premium principles. The first is the expected premium principle with a risk-loading factor $\theta$. The second is the distortion-based premium principle with a power distortion, which has an additional parameter $\rho$ compared to the expected premium principle. We then further generalize the premium principle to a fully flexible pricing kernel, which is a distortion function in a broad feasible class. On the demand side, we model farmer preferences using distortion risk measures that focus on a convex combination of Conditional Value-at-Risk (CVaR) and expected risk.

Our analysis contributes to the literature on Bowley solutions, which model a leader–follower game between a monopolistic insurer (leader) and a policyholder (follower). Under the Bowley solution, the insurance premium is first set by a chosen premium principle. Given this principle, the policyholder then selects its optimal coverage, which is a functional of this principle. The insurer knows this functional and chooses the premium principle that is optimal for her. In particular, this yields the insurer’s optimal pricing density for insurance contracts and the policyholder’s optimal insurance indemnity function \citep{chan1985reinsurer}. Prior works have investigated Bowley solutions under various preferences and pricing assumptions. Some of these have focused on using a profit-maximization objective for the insurer and a distortion risk-minimization objective for the policyholder. For example, \cite{cheung2019risk} investigate Bowley equilibria under distortion risk measures and general premium principles of a reinsurer (leader) and an insurer (follower). They derive the insurer’s explicit ceded loss and the reinsurer’s pricing functions via a two-step procedure that first minimizes the insurer’s distortion risk measure and then maximizes the reinsurer’s net gain. Both pricing and preferences are modeled through probability distortions, which are related to our use of distortion premium principles on the supply side and distortion risk measures on the demand side. \cite{chi2020bowley} revisit Bowley with a risk-neutral reinsurer and impose upper bounds on the first two moments of the indemnity. They further relax the assumptions of the expected premium principle in \cite{chan1985reinsurer} and the distortion premium principle in \cite{cheung2019risk}, and demonstrate that the monopoly premium can be determined without a specific premium principle. \cite{boonen2021bowley} also study Bowley solutions under the assumption that the reinsurer adopts a general premium function. A highlight is that their work analyzes Bowley reinsurance solutions under asymmetric information about the type of distortion function that the insurer adopts. The monopolistic reinsurer chooses a premium principle to maximize expected profit while anticipating the insurer’s type-dependent demand. Importantly, unlike the symmetric-information Bowley setting \citep[e.g.,][]{cheung2019risk} in which the leader can extract all surplus, asymmetric information prevents full extraction and allows both parties to strictly benefit from the reinsurance contracts. In addition, \cite{li2021bowley} study the sequential game under the mean-variance premium principle. In their setting, the buyer (follower) maximizes a mean-variance functional of its wealth and the seller (leader) maximizes its expected wealth by choosing the parameter of the mean-variance premium principle. Their findings demonstrate that the seller’s premium parameters change with respect to the buyer’s indemnity. There is also literature that adopts a different objective for the insurer (leader) and the policyholder (follower). For example, \cite{chen2024bowley} search for the set of Bowley solutions under the setup where the reinsurer minimizes the VaR measure of its risk position instead of maximizing its profit. Besides this, \cite{boonen2024bowley} analyze Bowley solutions with the policyholder’s objective replaced by expected-utility maximization. It shows that deductibles rise with the safety loading and proves that Bowley outcomes are Pareto dominated, which motivates policy intervention. In addition to this evidence, \cite{jiang2025bowley} study a Bowley insurance game with mean–variance preferences and a variance-based premium. They find the Bowley solution and show that it is never Pareto optimal, which supports the previous finding that such contracts are inefficient. This line of literature supports our game-theoretic approach in modeling Bowley equilibria. However, the current literature mainly focuses on finding Bowley solutions for indemnity insurance, where the contracts are written on the actual loss incurred. We adopt the concepts of Bowley solutions and apply them to index insurance. We investigate Bowley solutions for index insurance under both the expected premium principle and distortion premium principles with different distortion functions adopted by the policyholder.  

We also contribute to a growing actuarial literature that tackles basis risk in index insurance. Past literature has investigated reducing basis risk by using multivariate indices, flexible models to capture the complex relationships between indices and insurance payoff functions, as well as dependence modeling. \cite{zhang2019index} study optimal index insurance design under the expected-utility maximization framework. They find that the payoff functions can be highly nonlinear and nonmonotonic functions of the index variables, which align the payoff functions better with the actual loss and thus reduce basis risk. They highlight the generality of their model setup and its potential to be applied in various insurance applications. As a result, many research papers adopt the expected utility maximization framework for designing optimal weather index insurance. For example, \cite{tan2024flexible} use B-spline functions to define the feasible sets of the expected utility maximization problem and a penalty function to avoid overfitting. Besides this, \cite{chen2024managing} adopt the expected utility maximization framework in the design of weather index insurance and use neural networks (NNs) to model the insurance payoff functions. They propose to use neural networks to better capture the highly nonlinear relationships between the crop production losses and the high-dimensional weather indices. They show that the use of neural networks could reduce basis risk, lower insurance premiums, and improve the farmer's utility.  Our approach builds on theirs as we use convolutional neural networks to further reduce basis risk. Another common way to reduce basis risk is to construct lower-dimensional indices that provide stronger signals of weather conditions. For example, \cite{boyd2020design} use principal component regression and partial least-squares regression to construct multivariate indices. They show that these constructed indices perform better at correctly determining when payments should be triggered and at reducing the mismatch between the insurance payoff and the actual loss.  \cite{li2021improved} use a dynamic factor model for crop yield prediction, which also builds on the idea of summarizing high-dimensional weather variables using lower-dimensional latent factors and reducing dimensionality to improve prediction accuracy. The improved crop yield predictions could then help insurance pricing. \cite{zhu2024deep} also focuses on improving crop yield predictions by using factor models. They build a deep factor model with an encoder–decoder structure. The encoder compresses yields into a latent “production index” that captures nonlinear space–time patterns. They then feed weather and economic data into a concatenated network to model and predict this index. The decoder maps the predicted index back to crop yields. Besides this, \cite{zhu2018spatial} propose Lévy subordinated hierarchical Archimedean copulas, which produce a stronger fit for multiscale spatial dependence and joint extremes. This evidence supports designs that exploit spatiotemporal grids and reinforces the idea that CNN-based indices can be more robust under correlated shocks. This motivates our use of CNNs as an alternative way to model insurance payoff functions. In addition, hybrid product design has also evolved to balance moral hazard with basis risk. \cite{fan2023empirical} develop a model-based annealing random search and show that it effectively targets tail-risk objectives, which is crucial for lines exposed to catastrophes. Our penalty-based bilevel programming algorithm offers an alternative for large-scale, nonconvex design problems with equilibrium constraints, and fits settings where market power and pricing flexibility affect outcomes. 

We solve for the Bowley solutions in our settings using a penalized bilevel programming algorithm, which works even when the lower-level problem is not strongly convex \citep{shen2023penalty}. The problem is high-dimensional and not necessarily convex due to weather indices. We model the insurance payoff function with fully connected neural networks and compare it with the payoff functions modeled using CNNs, which preserve the index-by-time grid structure \citep{krizhevsky2012imagenet, chen2024managing}. We solve the two-stage game with a function-value-gap penalty. This links a penalized single-level problem to the original bilevel problem under mild assumptions. The result is a practical algorithm with convergence guarantees that scale to large weather–yield data.

Our main findings are as follows. First, the optimal monopoly insurance contract under CVaR consistently exhibits a stop-loss structure. Our model reproduces optimal indemnity contracts consistent with the literature \citep{cai2008optimal}; for index insurance, this shape persists, but contains more noise caused by basis risk. Second, CNN-based designs yield smoother, less noisy payoff functions than fully connected neural networks. We show that CNN-based designs are capable of reducing basis risk and pushing insurer profits closer to indemnity insurance levels. Third, insurer pricing responds sharply to farmer risk aversion: increasing the weight on CVaR in the farmer’s distortion measure leads to higher equilibrium loadings and greater profit extraction by the monopolist. Fourth, expanding the insurer’s pricing flexibility from a single risk loading to a two-parameter distortion premium, and then ultimately to a flexible pricing kernel, systematically increases equilibrium profits, which highlights the importance of premium regulation in concentrated markets.

The remainder of the paper proceeds as follows. Section \ref{sec:2} formulates the sequential game, introduces the farmer’s distortion risk measures, and the insurer’s premium principles. Section \ref{sec:3} provides theory on the penalized bilevel programming algorithm that we adopt in this study. Section \ref{sec:4} describes the data and neural-network architectures for index insurance design. Section \ref{sec:5} reports numerical results, including sensitivity analysis, comparisons of fully connected networks and CNNs, and extensions to two-parameter and general pricing-kernel settings. Section \ref{sec:6} concludes with policy implications and directions for future research.

\section{Problem Formulation}\label{sec:2}
\subsection{General Setup}
Let $(\Omega,\mathcal F,\mathbb{P})$ be a probability space such that $\Omega$ is finite and $\mathcal F=2^\Omega$ is the power $\sigma$-algebra.\footnote{This paper considers empirical distributions using numerical methods, and so we use a simplified setting with finitely many states of the world. Many results can be extended to infinite state spaces under standard integrability assumptions, as long as the strategy spaces are finite dimensional.} Moreover, denote by $\mathbb{R}^\Omega_+$ the set of nonnegative random variables on $(\Omega,\mathcal F,\mathbb{P})$.

Consider a farmer seeking to hedge against weather-related losses. The farmer is faced with a stochastic production loss $Y$. We construct an index-insurance contract, where the insurance payoff is determined by a function of $p$-dimensional vector of weather indices $\mathbf{X} = (X_1, X_2, ..., X_p)$. The nonnegative insurance payoff function is denoted by $I(\mathbf{X})$ with $ I: \mathbb{R}^p \rightarrow \mathbb{R}_+$. The premium is denoted by $\Pi(I(\mathbf{X}))$, which is a function of $I(\mathbf{X})$. 

\subsection{Distortion Function}
Both the farmer and the insurer use a transformation of probabilities to evaluate and price random losses. This is used to formulate the optimization problem, 
Let $\mathcal G$ be the class of continuous and nondecreasing functions $g:[0,1] \rightarrow\mathbb{R}_+$ such that $g(0)=0$. For $g\in\mathcal G$, we can write the distortion risk measure for risk $Z\in\mathbb{R}_+^\Omega$ as:
\begin{equation}\label{eq:premium}\int_0^{\infty} g(\mathbb{P}(Z>z)) dz=\sum_{k=1}^{q-1} g\left(\mathbb{P}\left(\left\{\omega_1, \ldots, \omega_k\right\}\right)\right) \cdot\left[Z\left(\omega_k\right)-Z\left(\omega_{k+1}\right)\right]+Z\left(\omega_q\right), \end{equation} if the state space $\Omega=\left\{\omega_1, \ldots, \omega_q\right\}$ is such that $Z\left(\omega_1\right) \geq \cdots \geq Z\left(\omega_q\right)$; see, e.g., \cite{de2008second} and \cite{boonen2015competitive}. We use premium principles of the form \eqref{eq:premium}, and we denote such principle as $\hat \Pi(Z)$. As special cases, we will study the case where $g$ is linear or when $g(s)=(1+\theta)s^{1/\rho}$ for $\theta\geq0$ and $\rho\geq 1$. 

Let $\rho_F$ be the distortion risk measure of the farmer, and the distortion risk measure of a loss $Z\in\mathbb{R}^\Omega_+$ is given by
\begin{equation*}\label{eq:rho_f}
    \rho_F(Z)=\int_0^{\infty}g_f(\mathbb{P}(Z>z)) dz,
\end{equation*}
where $g_f$ is the distortion function of the farmer. A distortion function is a nondecreasing and concave function $g:[0,1] \rightarrow[0,1]$ such that $g(0)=0$. The class of distortion functions is denoted by $\mathcal G_d$, and, clearly, $\mathcal G_d\subset \mathcal G$. 

The assumption of concavity of the distortion function implies aversion to mean-preserving spreads of the corresponding distortion risk measure \citep{yaari1987}.
The distortion functions employed in this work for the farmer's risk measure include: 1. Conditional Value-at-Risk (CVaR). 2. A convex combination of CVaR and expected risk.

\subsubsection{Conditional Value at Risk}
Conditional Value-at-Risk (CVaR), also known as Expected Shortfall, quantifies the expected value of the loss given that the loss exceeds a specified threshold. Mathematically, for a random variable $Z$ and a confidence level $\alpha \in (0, 1)$, CVaR is defined as:
\begin{equation}
    \text{CVaR}_{\alpha}(Z) = \frac{1}{1-\alpha}\int_\alpha^1 \text{VaR}_{u}(Z)du,
\end{equation}
where $\text{VaR}_{\alpha}(Z)$ represents the Value-at-Risk at level $\alpha$, which is defined as the $\alpha$-th quantile of the distribution of $Z$:

\begin{equation*}
    \text{VaR}_{u}(Z) = \inf \{ x \in \mathbb{R} : P(Z \leq x) \geq u \}, u\in[0,1].
\end{equation*}
For continuous random variables, the CVaR is equal to 
\begin{equation}
    \text{CVaR}_{\alpha}(Z) =E[Z|Z\geq \text{VaR}_{\alpha}(Z)].
\end{equation}
The risk measure $\text{CVaR}_\alpha$ is a distortion risk measure with distortion function $g(s)=\min\{\frac{s}{1-\alpha},1\}$ for $s\in[0,1]$ \citep{dhaene2006}.
\subsubsection{Convex Combination of CVaR and Expected Risk}
To manage the trade-off between risk and expected return, we incorporate a convex combination of CVaR and expected risk. This combination is governed by a weight $\lambda \in [0,1]$, which determines the relative importance of CVaR versus expected risk. The convex combination is expressed as $\rho(Z) = \lambda\,\mathbb{E}(Z) + (1-\lambda)\,\operatorname{CVaR}_{\alpha}(Z)$, and the corresponding distortion function is given by $g(s) = \lambda s + (1-\lambda)\min\{s/(1-\alpha), 1\}$ for $s\in[0,1]$. This places $\rho$ within the distortion-risk-measure framework, yielding an increasing, concave distortion with a kink at  $s = 1-\alpha$ \citep{cheung2017characterizations}. 

The distortion function $g(s)$ is largely linear in the body of the distribution, while adding extra weight to the worst $1-\alpha$ tail. In the tail region ($s \le 1-\alpha$), the slope is $\lambda + (1-\lambda)/(1-\alpha)$, which upweights rare and severe losses. In the body ($s > 1-\alpha$), the slope is $\lambda$, recovering mean-like behavior. Thus, $\lambda$ controls the appetite for risk: $\lambda \to 1$ makes $\rho$ approach the mean, while $\lambda \to 0$ makes $\rho$ approach pure CVaR. The parameter $\alpha$ sets the tail horizon: larger $\alpha$ focuses on rarer extremes. Because $g$ is concave, the associated distortion risk measure is coherent \citep{artzner1999} and, in particular, satisfies subadditivity \citep{wang1997}.

\subsubsection{Power function}
The power distortion function $g(s)=s^{1/\rho}$ is used to embed a transparent, one-parameter distortion function into insurance and reinsurance pricing by overweighting or underweighting tail probabilities in a controlled way. The parameter $\rho$ governs the shape and thus the attitude toward risk: when $\rho>1$, $1/\rho<1$, and $g$ is concave on $[0,1]$, so $g(s)\ge s$ for $s\in(0,1)$. This inflates survival probabilities, raises premiums, and represents a risk-averse tail loading. When $\rho<1$, $1/\rho>1$, and the function $g$ is convex, so $g(s)\le s$. This downweights tail probabilities, lowers premiums, and corresponds to risk-seeking behavior. At $\rho=1$, $g$ is linear, and this yields risk-neutral pricing. This distortion is also known as the proportional hazards transform. This distortion function links the premiums to tail risk, which aligns with market-consistent pricing ideas. Compared with Wang’s transform, which distorts probabilities via the (standard) normal CDF, the power form is a simpler alternative within the same distortion-based framework and is convenient when a single parameter, a closed form, and tail emphasis are desired \citep[see, e.g.,][]{wang1996premium}. 

\subsection{Sequential game}
The farmer and the insurer make decisions sequentially: the insurer selects the pricing rule $\Pi$ first, and the farmer then chooses the insurance payoff function $I(\mathbf{X})$ after observing the pricing rule. The farmer chooses $(\mathbf{X})$ to minimize its risk as measured by $\rho_F$. The insurer chooses the premium function $\Pi(I(\mathbf{X}))$ to maximize its profit, taking into account the farmer’s decisions on the payoff functions for each possible premium function. We express this sequential decision-making problem as a bilevel optimization problem, with the insurer’s objective as the upper problem and the farmer’s objective as the lower problem.\\


\noindent\textbf{Upper Problem (UP): Insurer}
\begin{equation*}
 \operatorname{Max}_{\Pi \in \mathcal{P}_s} \Pi(I(\mathbf{X}))-(1+\mu) \mathbb{E}(I(\mathbf{X})), \end{equation*}
where the feasible set $\mathcal P_s$ is a closed and convex subset of the following set:
\begin{equation*} \mathcal{P}=\left\{\hat{\Pi}: \mathbb{R}^\Omega_+ \rightarrow \mathbb{R}_{+} \mid \hat{\Pi}(I(\mathbf{X}))= \int_0^{\infty} g_i(\mathbb{P}(I(\mathbf{X})>z)) dz, g_i\in\mathcal G\right\}.\end{equation*}
\textbf{Lower Problem (LP): Farmer} 
\begin{equation}\min _{I \in \mathcal{I}} \rho_F\left(Y - I(\mathbf{X}) + \Pi(I(\mathbf{X}))\right).\label{eq:LP}
\end{equation}\medskip

As special cases of $\mathcal P_s$, we consider
\begin{equation*}
\begin{aligned}
&\mathcal{P}_{s1}=\left\{\hat{\Pi}: \mathbb{R}^\Omega_+ \rightarrow \mathbb{R}_{+} \mid \hat{\Pi}(I(\mathbf{X}))=(1+\theta) \mathbb{E}(I(\mathbf{X})), \theta \geq 0\right\},\text { or} \\
&\mathcal{P}_{s2}=\left\{\hat{\Pi}: \mathbb{R}^\Omega_+ \rightarrow \mathbb{R}_{+} \mid \hat{\Pi}(I(\mathbf{X}))=(1+\theta) \int_0^{\infty} \mathbb{P}(I(\mathbf{X})>z)^{\frac1\rho} dz,\theta\geq0, \rho \geq 1\right\}. 
\end{aligned}
\end{equation*}
We denote the first special case of $\mathcal P_{s}$ as $\mathcal P_{s1}$ (Problem 1), where we explicitly set the insurer's pricing rule to the expected premium principle. This means that the insurer selects only the risk-loading factor $\theta$ in order to maximize its profit. The second special case of $\mathcal P_{s}$ is denoted by $\mathcal P_{s2}$ (Problem 2). In this special case, we adopt the power distortion function for the premium, and the insurer then maximizes its profit by choosing two parameters $\rho$ and $\theta$. In this case, the insurer has more flexibility in its pricing strategy compared with $\mathcal P_{s1}$. We begin by analyzing solutions to these two special cases (Problems 1 and 2). We then generalize the premium principle to the feasible set $\mathcal{P}_{s} = \mathcal{P} $, and denote this setting as Problem 3, in which we investigate equilibrium solutions under the general pricing function $g_i$.

\section{Model}\label{sec:3}
\subsection{Neural Network Based Models}
The deep-learning models that we propose to obtain the insurance payoffs are fully connected neural networks (NNs) and convolutional neural networks (CNNs). CNNs are a special class of neural networks that incorporate convolution in at least one of their layers. In this study, we compare these two classes of models. NNs process flattened vector inputs; thus, the two-dimensional time-index matrices must be flattened. CNNs are particularly suitable for processing grid-structured data without destroying its two-dimensional structure. This makes CNNs a well-motivated alternative to the traditional fully connected neural networks proposed in \cite{chen2024managing}.  

\subsubsection{Fully Connected Neural Networks (NNs)}
Neural Networks (NNs) are models composed of interconnected neurons that are organized in layers. Each neuron processes inputs, applies a weight, adds a bias, and passes the result through an activation function. Specifically, the outputs of neurons in one layer serve as inputs to neurons in the subsequent layer. The most common type of NN is the fully connected architecture, which is illustrated in Figure \ref{fig:nn_illustration}. This model receives as input an $n$-dimensional vector, $X = (X_1, X_2, \dots, X_n)^T$, and outputs a single prediction $y$.

In a deep learning NN architecture, the structure consists of an input layer, an output layer, and $D$ hidden layers ($D > 1$). The $d$-th ($d = 1, 2, \dots, D$) hidden layer, $Z^{(d)} = (Z^{(d)}_1, Z^{(d)}_2, \dots, Z^{(d)}_{n_d})^T$, contains $n_d$ neurons. Each neuron in the hidden layer is obtained by applying a nonlinear function, $\sigma_{d-1}$, to the linear combination of the neurons from the previous layer. This is formally expressed as:

\begin{equation}
\mathbf{Z}^{(d)} = \sigma_{d-1} (\mathbf{\alpha}^{(d-1)} + \mathbf{W}^{(d-1)} \mathbf{Z}^{(d-1)}),
\end{equation}
where:
\begin{itemize}
\item $W^{(d-1)}$ is an $(l_h \times l_{h-1})$-dimensional weight matrix;
\item $\alpha^{(d-1)}$ is a $(n_d \times 1)$-vector of bias units, capturing intercepts in the model;
\item $\sigma_{h-1}$ is the activation function, which is typically nonlinear and predefined.

\end{itemize}

This fully connected structure ensures that neurons between two adjacent layers are pairwise connected.

\begin{figure}
    \centering
    \includegraphics[width=1\textwidth]{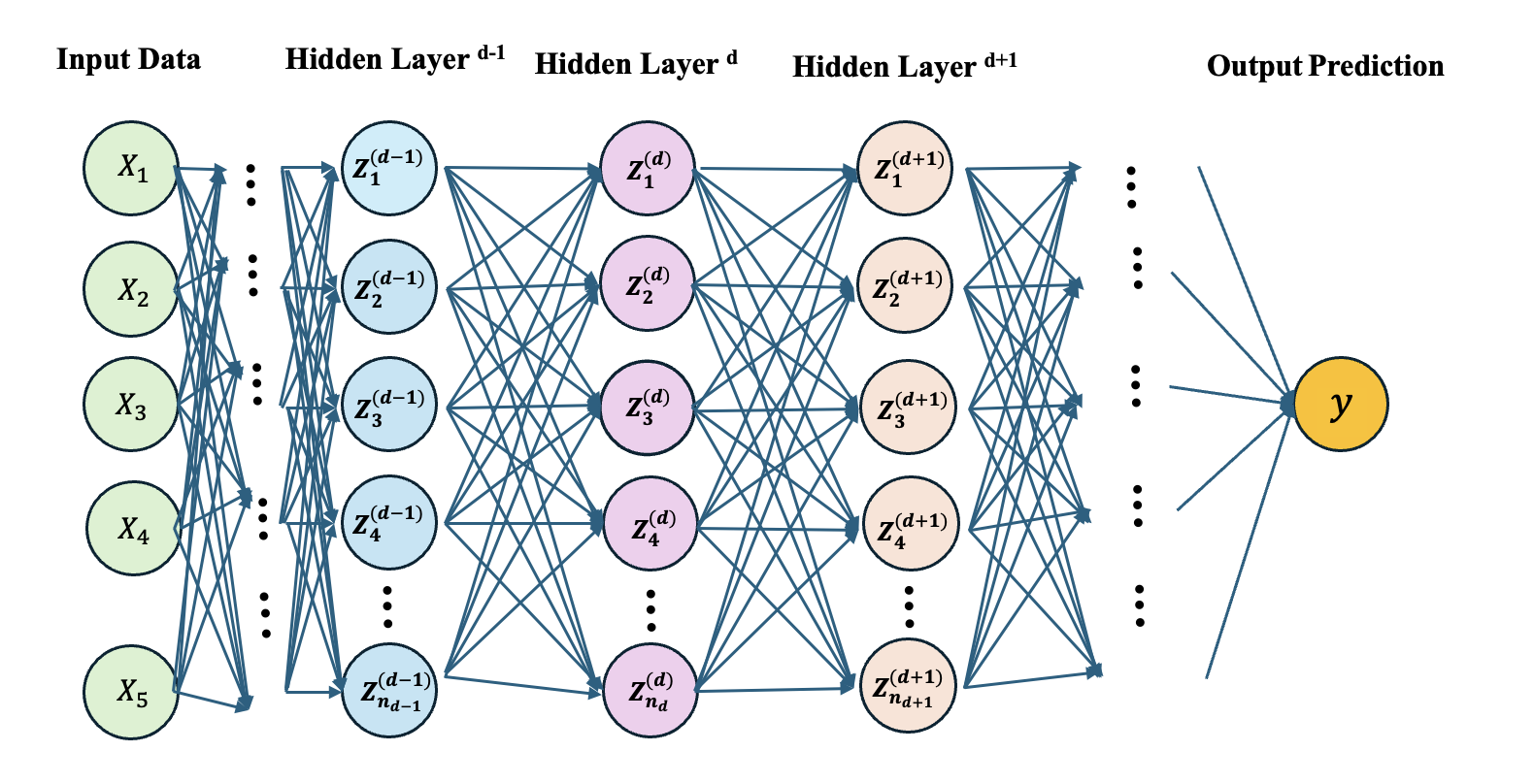}
    \caption{Illustration of NN structure.}
    \label{fig:nn_illustration}
\end{figure}

\subsubsection{Convolutional Neural Networks (CNNs)}
CNNs are a type of deep learning model primarily used for processing structured grid-like data such as images. They automatically and adaptively learn spatial hierarchies of features through layers of convolution, pooling, and fully connected layers. After the breakthrough work by \cite{krizhevsky2012imagenet}, CNNs became well-known for deep learning models with a larger volume of data. 

CNNs offer several key advantages that make them highly effective for handling structured data. First, they use convolutional kernels to efficiently extract local features, which reduces computational complexity and the number of parameters compared to fully connected networks. Second, CNNs adopt weight sharing, where the same set of weights is applied across different regions of the input. This reduces the number of parameters that need to be learned and ensures that the model is robust to small translations or shifts in the input, which is particularly useful for images where features often appear in various locations. Third, CNNs also learn features hierarchically, allowing them to recognize simple patterns like edges in early layers and complex structures like objects in deeper layers. This hierarchical learning allows CNNs to capture both local and global patterns in the data, making them highly effective for complex data. For our study, CNNs are particularly beneficial for processing the index-time matrix as images. The index-time matrix can be reshaped into image-like structures, where each pixel represents a specific location in space and time. Treating the index-time matrix as an image allows CNNs to leverage their unique capabilities. 

The construction of CNNs involves four main components: (1) convolutional operation; (2) activation; (3) pooling; and (4) fully connected neural networks. In the convolutional layer, units are organized into feature maps that detect local patterns in the data. These feature maps then pass through a nonlinear activation function.  After activation, a max pooling layer is applied to downsample the feature maps by selecting the maximum value within each region. This step reduces the number of parameters, which helps prevent overfitting and makes the model more robust to small shifts in the data. In the final stage of the architecture, the outputs from the convolution, activation, and pooling layers are flattened and combined into fully connected neural-network layers, which then generate the model's forecast of the insurance payoff $I(X)$. Note that the insurance payoffs cannot be negative. Therefore, we use the ReLU (Rectified Linear Unit) activation function ($\sigma(x) = max(x, 0)$) in our model to ensure that the predicted insurance payoffs are nonnegative. The illustration of the CNN structures adopted in this study is shown in Figure \ref{fig: cnn_illustration}.

\begin{figure}
    \centering
    \includegraphics[width=1\textwidth]{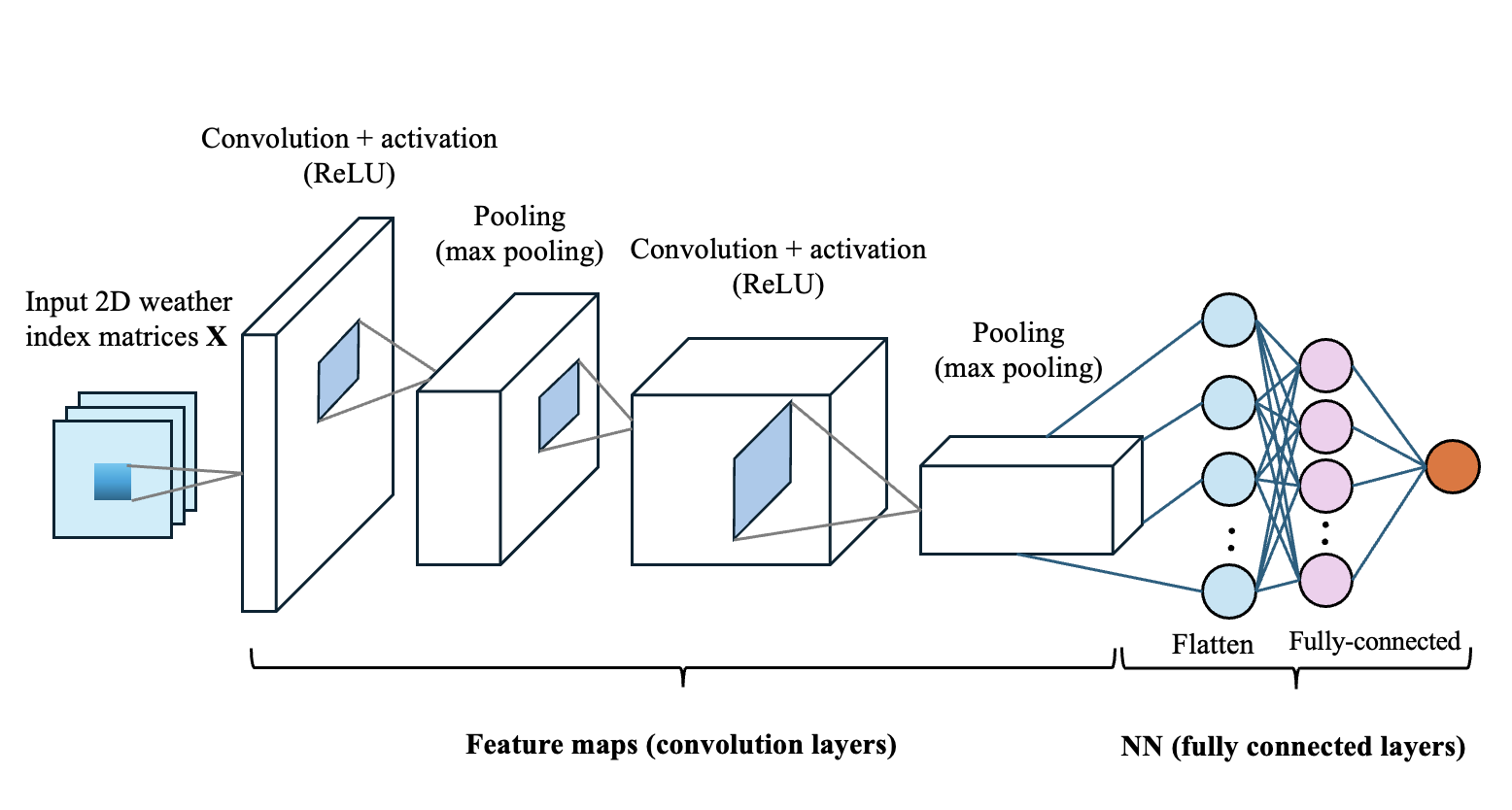}
    \caption{Illustration of CNN structure.}
    \label{fig: cnn_illustration}
\end{figure}

\subsection{Penalized Bilevel Programming}
Bilevel optimization plays an increasingly important role in machine learning \cite{shen2023penalty} and is particularly well-suited for problems with a sequential structure where one level of decision is nested within another. In our problem formulation, the insurer's decision (from solving the UP problem) depends on the outcome of the farmer's LP problem. This makes bilevel programming algorithms promising for our study.

In general, the structure of a bilevel programming problem can be expressed as:

\begin{equation}
\mathcal{BP}: \min_{x, y} f(x, y) \text{ s.t. } x \in \mathcal{C}, y \in \mathcal{S}(x) := \arg \min_{y \in \mathcal{U}(x)} h(x, y),
\end{equation}
where $\mathcal{C} \subseteq \mathbb{R}^{q}$ is a nonempty, closed set, and $\mathcal{U}(x)$ and $\mathcal{S}(x)$ are nonempty, closed sets given any $x \in \mathcal{C}$. The functions $f$ and $h$ are referred to as the upper-level and lower-level objectives, respectively. 

The bilevel optimization problem $\mathcal{BP}$ is inherently challenging due to the coupling between the upper-level and lower-level problems through the solution set $\mathcal{S}(x)$. In particular, solving $\mathcal{BP}$ can be extremely difficult when the lower-level function $h(x, \cdot)$ is not strongly convex or is constrained, as traditional implicit gradient methods are unable to handle such cases directly. To address these challenges, recent work has explored penalty-based reformulations of $\mathcal{BP}$, which transform the bilevel problem into a single-level optimization problem by penalizing certain optimality metrics of the lower-level problem.

The penalty reformulation of $\mathcal{BP}$ is defined as:

\begin{equation}
\mathcal{BP}_{\gamma p}: \min_{x, y}  F_\gamma(x, y) = f(x, y) + \gamma p(x, y) \text{ s.t. } x \in \mathcal{C}, y \in \mathcal{U}(x),
\end{equation}
where $p(x, y)$ is a penalty function that measures the distance to the lower-level solution set $\mathcal{S}(x)$. We define a squared-distance-bound function to ensure that $p(x, y)$ approximates the distance from y to $\mathcal{S}(x)$.

\begin{definition}
(Squared-distance-bound function). A function $p: \mathbb{R}^{d_x} \times \mathbb{R}^{d_y} \mapsto \mathbb{R}$ is a $\rho$-squared-distance-bound function if there exists $\rho > 0$ such that for any $x \in \mathcal{C}, y \in \mathcal{U}(x)$, it holds:
\begin{equation}
\begin{aligned}
& p(x, y) \geq 0, \quad \rho \cdot p(x, y) \geq d_{\mathcal{S}(x)}^2(y), \\
& p(x, y) = 0 \text{ if and only if } d_{\mathcal{S}(x)}(y) = 0,
\end{aligned}
\end{equation}
where $d_{\mathcal{S}(x)}(y)$ is the Euclidean distance from $y$ to the lower-level solution set $\mathcal{S}(x)$.
\end{definition}

Given a squared-distance-bound function $p(x, y)$, we can define an $\epsilon$-approximate version of the bilevel problem $\mathcal{BP}$:

\begin{equation}
\mathcal{BP}_\epsilon: \min_{x, y} f(x, y) \text{ s.t. } x \in \mathcal{C}, y \in \mathcal{U}(x), p(x, y) \leq \epsilon.
\end{equation}
When $\epsilon = 0$, $\mathcal{BP}_\epsilon$ recovers the original bilevel problem $\mathcal{BP}$. For $\epsilon > 0$, $\mathcal{BP}_\epsilon$ serves as an $\epsilon$-approximation of $\mathcal{BP}$, where the lower-level solution $y$ is within $\epsilon$ of the true lower-level solution set $\mathcal{S}(x)$.

\cite{shen2023penalty} show that, under certain generic conditions on $h(x, y)$ that hold without the strong convexity of $h(x,y)$, global (local) solutions of $\mathcal{BP}$ can be approximated by solving $\mathcal{BP}_{\gamma p}$. We follow their approach to establish the relationship between $\mathcal{BP}$, $\mathcal{BP}_\epsilon$, and $\mathcal{BP}_{\gamma p}$.

\begin{definition}
(Lipschitz continuity). Given $L > 0$, a function $\ell: \mathbb{R}^d \mapsto \mathbb{R}^{d'}$ is said to be $L$-Lipschitz continuous on $\mathcal{X} \subseteq \mathbb{R}^d$ if for any $x, x' \in \mathcal{X}$, it holds:
\begin{equation}
\|\ell(x) - \ell(x')\| \leq L \|x - x'\|.
\end{equation}
\end{definition}

\begin{assumption}[Lipschitz continuity of upper objective]
\label{ass:lipschitz_f}
There exists a constant $L$ such that for any $x \in \mathcal{C}$, the upper-level objective $f(x, \cdot)$ is $L$-Lipschitz continuous on $\mathcal{U}(x)$.
\end{assumption}

According to \eqref{eq:premium}, our UP problem can be simplified to a linear function of order statistics, which is Lipschitz continuous by linearity. Thus, Assumption \ref{ass:lipschitz_f} is satisfied directly. The key result that establishes the relationship between $\mathcal{BP}$, $\mathcal{BP}_{\gamma p}$, and $\mathcal{BP}_\epsilon$ is the following theorem, which shows that global and local solutions of the penalized problem $\mathcal{BP}_{\gamma p}$ correspond to those of the original bilevel problem $\mathcal{BP}$.

\begin{theorem}[General Relation on Global Solutions]
\label{thm:gen_global}
Assume $p(x, y)$ is a $\rho$-squared-distance-bound function and Assumption \ref{ass:lipschitz_f} holds. Given any $\epsilon_1 > 0$, any global solution of $\mathcal{BP}$ is an $\epsilon_1$-global-minimum point of $\mathcal{BP}_{\gamma p}$ with any $\gamma \geq \gamma^* = \frac{L^2 \rho}{4} \epsilon_1^{-1}$. Conversely, given $\epsilon_2 \geq 0$, if $(x_\gamma, y_\gamma)$ achieves an $\epsilon_2$-global-minimum of $\mathcal{BP}_{\gamma p}$ with $\gamma > \gamma^*$, $(x_\gamma, y_\gamma)$ is the global solution of $\mathcal{BP}_{\epsilon_\gamma}$ with some $\epsilon_\gamma \leq \frac{\epsilon_1 + \epsilon_2}{\gamma - \gamma^*}$.
\end{theorem}

Theorem \ref{thm:gen_global} implies that solving the penalized problem $\mathcal{BP}_{\gamma p}$ can yield an approximate solution to the original bilevel problem $\mathcal{BP}$. The parameter $\gamma$ controls the trade-off between the upper-level objective and the penalty term, which ensures that the solution remains close to the lower-level solution set.

\begin{theorem}[General Relation on Local Solutions]
\label{thm:gen_local}
Assume $p(x, \cdot)$ is continuous for any $x \in \mathcal{C}$ and $p(x, y)$ is a $\rho$-squared-distance-bound function. Given $\gamma > 0$, let $(x_\gamma, y_\gamma)$ be a local solution of $\mathcal{BP}_{\gamma p}$ on $\mathcal{N}((x_\gamma, y_\gamma), r)$, where $\mathcal{N}$ is a circle with a center of $(x_\gamma, y_\gamma)$ and a radius $r$. Assume $f(x_\gamma, \cdot)$ is $L$-Lipschitz continuous on $\mathcal{N}(y_\gamma, r)$. Assume either one of the following is true:
(i) There exists $\bar{y} \in \mathcal{N}(y_\gamma, r)$ such that $\bar{y} \in \mathcal{U}(x_\gamma)$ and $p(x_\gamma, \bar{y}) \leq \epsilon$ for some $\epsilon \geq 0$. Define $\bar{\epsilon}_\gamma = \frac{L^2 \rho}{\gamma^2} + 2\epsilon$.
(ii) The set $\mathcal{U}(x_\gamma)$ is convex and the function $p(x_\gamma, \cdot)$ is convex. Define $\bar{\epsilon}_\gamma = \frac{L^2 \rho}{\gamma^2}$.

Then $(x_\gamma, y_\gamma)$ is a local solution of $\mathcal{BP}_{\epsilon_\gamma}$ with $\epsilon_\gamma \leq \bar{\epsilon}_\gamma$.
\end{theorem} 

Theorem \ref{thm:gen_local} provides conditions under which a local solution of the penalized problem corresponds to a local solution of the original bilevel problem.

To apply this framework, we select the value-gap penalty, $p(x, y)=h(x, y)-v(x)$, where $v(x) := \min_{y' \in \mathcal{U}(x)} h(x, y')$. In our insurance game, the lower-level objective $h(x,\cdot)$ based on CVaR is polyhedral and convex, but not strongly convex. Therefore, we leverage the property of a weak sharp minimum, which holds for polyhedral convex functions on compact sets.

\begin{assumption}[Weak Sharp Minimum]
\label{ass:weak_sharp_min}
The lower-level objective $h(x, \cdot)$ has the property of a weak sharp minimum on the compact set $\mathcal{U}(x)$. That is, there exists a constant $\kappa > 0$ such that for any $x \in \mathcal{C}$ and $y \in \mathcal{U}(x)$:
\begin{equation}
h(x, y) - v(x) \ge \kappa \cdot d_{\mathcal{S}(x)}(y).
\end{equation}
\end{assumption}

This assumption allows us to prove that our chosen penalty function satisfies the squared-distance-bound condition required by Theorems \ref{thm:gen_global} and \ref{thm:gen_local}. According to \eqref{eq:premium}, our LP objective \eqref{eq:LP} can be simplified to a linear function of order statistics, which has a weak sharp minimum if $\mathcal{U}(x)$ is nonempty. Thus, Assumption \ref{ass:weak_sharp_min} is satisfied directly.

\begin{lemma}[Value-Gap as a Squared-Distance Bound]
\label{lem:sq_dist_bound}
Under Assumption \ref{ass:weak_sharp_min}, if the feasible set $\mathcal{U}(x)$ is compact with a uniform diameter $D$, then the value-gap penalty $p(x, y)=h(x, y)-v(x)$ is a $\rho$-squared-distance-bound with $\rho = D/\kappa$.
\end{lemma}
\begin{proof}
By definition, $d_{\mathcal{S}(x)}^2(y) = d_{\mathcal{S}(x)}(y) \cdot d_{\mathcal{S}(x)}(y)$. Since $y \in \mathcal{U}(x)$ and $\mathcal{S}(x) \subset \mathcal{U}(x)$, the distance $d_{\mathcal{S}(x)}(y)$ is bounded by the diameter $D$ of $\mathcal{U}(x)$. Thus, $d_{\mathcal{S}(x)}^2(y) \le D \cdot d_{\mathcal{S}(x)}(y)$. From Assumption \ref{ass:weak_sharp_min}, we have $d_{\mathcal{S}(x)}(y) \le \frac{1}{\kappa}(h(x,y) - v(x))$. Substituting this yields:
$d_{\mathcal{S}(x)}^2(y) \le \frac{D}{\kappa} (h(x,y) - v(x))$. Setting $\rho = D/\kappa$ and $p(x, y) = h(x, y) - v(x)$, we have $\rho \cdot p(x, y) \ge d_{\mathcal{S}(x)}^2(y)$. The other conditions follow from the definitions.
\end{proof}

By directly applying Theorems \ref{thm:gen_global}-\ref{thm:gen_local} with $\rho=D/\kappa$, we get the following specific results for our problem.

\begin{proposition}[Relation on Global Solutions for the Insurance Game]
\label{prop:global_sol}
Assume Assumptions \ref{ass:lipschitz_f} and \ref{ass:weak_sharp_min} hold. Let $p(x, y)=h(x, y)-v(x)$. Suppose $\gamma \geq L \sqrt{(\rho/2) \delta^{-1}}$ with $\rho = D/\kappa$ for some $\delta > 0$. If $(x_\gamma, y_\gamma)$ is a global solution of $\mathcal{BP}_{\gamma p}$, then it is a global solution of $\mathcal{BP}_{\epsilon_\gamma}$ with $\epsilon_\gamma \leq \delta$.
\end{proposition}

\begin{proposition}[Relation on Local Solutions for the Insurance Game]
\label{prop:local_sol}
Assume Assumptions \ref{ass:lipschitz_f} and \ref{ass:weak_sharp_min} hold. For some $\delta > 0$, let $p(x, y) = h(x, y) - v(x)$ and $\gamma \geq L \sqrt{(3\rho/2) \delta^{-1}}$ with $\rho = D/\kappa$. If $(x_\gamma, y_\gamma)$ is a local solution of $\mathcal{BP}_{\gamma p}$, it is also a local solution of $\mathcal{BP}_{\epsilon_\gamma}$ with $\epsilon_\gamma \leq \delta$.
\end{proposition}

Propositions \ref{prop:global_sol} and \ref{prop:local_sol} show the relationship between the penalized problem $\mathcal{BP}_{\gamma p}$ and the $\epsilon$-approxi-\\
mation $\mathcal{BP}_{\epsilon_\gamma}$. Specifically, the optimal solution of the penalized problem $\mathcal{BP}_{\gamma p}$ is a tight approximation for the original problem $\mathcal{BP}$. We use a stochastic gradient descent method to optimize $\mathcal{BP}_{\gamma p}$. Due to the nonsmoothness of the objectives (linear order statistics \eqref{eq:premium}), we need to use the subgradient method based on Danskin's Theorem. 

\begin{lemma}[Danskin's Theorem for Subgradients]
\label{lem:danskin}
If $h(x, y)$ is continuous, convex in $y$ for each fixed $x$, and continuously differentiable in $x$ for each fixed $y$, and the set $\mathcal{U}(x)$ is compact, then the value function $v(x) = \min_{y \in \mathcal{U}(x)} h(x, y)$ is directionally differentiable. A subgradient of $v(x)$ can be computed as:
\begin{equation}
\nabla v(x) \in \text{conv} \left\{ \nabla_x h(x, y^*) \mid y^* \in \mathcal{S}(x) \right\}.
\end{equation}
In particular, for any single solution $y^* \in \mathcal{S}(x)$, the vector $\nabla_x h(x, y^*)$ is a valid subgradient of $v(x)$.
\end{lemma}

This justifies our algorithm, V-PBGD (Algorithm \ref{algo:v-pbgd}), which approximates this by first finding an approximate lower-level solution $\hat{y}^k$ and then using $\nabla_x h(x^k, \hat{y}^k)$ as an estimator for $\nabla v(x^k)$.

The core of our approach is to apply a gradient-based method to the penalized objective function $F_\gamma(x, y)$. This requires computing the gradient of the penalty term $p(x, y) = h(x, y) - v(x)$, which in turn requires a value for $\nabla v(x)$, the gradient of the lower-level value function.

A key component of this gradient is $\nabla v(x)$. As discussed, the nonsmooth nature of $h(x,y)$ prevents the use of standard envelope theorems. Instead, we rely on Danskin's Theorem. For our problem, Danskin's Theorem states that a subgradient of the value function $v(x)$ can be computed as $\nabla_x h(x, y^*)$, where $y^*$ is any optimal solution to the lower-level problem $\min_y h(x, y)$.

This theoretical foundation allows us to formulate a practical algorithm. In practice, computing the exact $y^*$ at each step is computationally expensive. Therefore, for outer iteration $k$ and given $x^k$, we perform a fixed number $T_k$ of inner steps of subgradient descent to approximate the solution to the lower-level problem:

\begin{equation}
    \omega^{(k)}_{t+1} = \omega^{(k)}_t - \beta \, \partial_y h(x^k, \omega^{(k)}_t), \quad t = 1, \ldots, T_k,
\end{equation}
with the inner loop initialized with $\omega^{(k)}_1 = y^k$. Note that we use $\partial_y h$ to denote a subgradient, as the lower-level objective is nonsmooth. This process yields an approximate lower-level solution $\hat{y}^k = \omega^{(k)}_{T_k + 1}$. We can then approximate the full penalized gradient at $(x^k, y^k)$ by using this $\hat{y}^k$ to serve as a proxy for the true minimizer $y^*$. The complete update for $(x^k, y^k)$ is as follows:

\begin{equation}
    (x^{k+1}, y^{k+1}) = \text{Proj}_Z \left( (x^k, y^k) - \alpha \left(\nabla f(x^k, y^k) + \gamma (\partial_y h(x^k, y^k) - \nabla_x h(x^k, \hat{y}^k))\right) \right),
\end{equation}
where $Z = \mathcal{C} \times \mathbb{R}^{d_y}$ is the feasible set for the variables, $\alpha$ is the learning rate, and the gradient term $\nabla_x h(x, \hat{y}^k)$ is interpreted as a vector in the joint $(x, y)$ space, with zeros in the $y$ coordinates. This iterative update process is summarized in Algorithm 1, which we refer to as V-PBGD (Value-function-based Penalty Bilevel Gradient Descent).

\begin{algorithm}
\caption{V-PBGD: Function value gap-based fully first-order PBGD }\label{algo:v-pbgd}
\begin{algorithmic}
\State Select $(x^1, y^1) \in Z = \mathcal{C} \times \mathbb{R}^{d_y}$. Select step sizes $\alpha$, $\beta$, constant $\gamma$, iteration numbers $T_k$, and $K$.
\For{$k = 1$ to $K$}
    \State Initialize $\omega^{(k)}_1 = y^k$.
    \For {$t = 1$ to $T_k$}
        \State Update $\omega^{(k)}_{t+1} = \omega^{(k)}_t - \beta \nabla_y h(x^k, \omega^{(k)}_t)$.
    \EndFor
    \State Set $\hat{y}^k = \omega^{(k)}_{T_k + 1}$.
    \State Use $\hat{y}^k$ to approximate $\nabla v(x^k)$ via $\nabla_x h(x^k, \hat{y}^k)$.
    \State Update $(x^{k+1}, y^{k+1}) = \text{Proj}_Z \left( (x^k, y^k) - \alpha (\nabla f(x^k, y^k) + \gamma (\nabla_y h(x^k, y^k) - \nabla_x h(x^k, \hat{y}^k))) \right)$.
\EndFor
\State Output $(x^K, y^K)$
\end{algorithmic}
\end{algorithm}

This algorithm effectively manages the computational challenges of bilevel optimization by utilizing the function value-gap penalty and iteratively resolving the lower-level problem using gradient descent. The convergence analysis of the subgradient method can be found in Theorem 2 of \cite{shamir2013stochastic}

We follow the construction of the function value-gap penalty $p(x, y)=h(x, y)-v(x)$ and evaluate the penalized gradient using $\nabla v(x) = \nabla_x h(x, y^*)$, where $y^* \in \mathcal{S}(x)$ is the lower-level problem solution. We then propose the following combined single objective for our first bilevel problem (feasible set $\mathcal P_{s1}$) in searching for optimal insurance strategies:

\begin{equation}
\begin{aligned}
&\min _{\theta, I^{\theta}(\mathbf{X})} -\left( \Pi(I^\theta(\mathbf{X})) - (1 + \mu)\mathbb{E}\left(I^\theta(\mathbf{X})\right) \right) + \gamma \cdot \left( \text{LP}(I^{\theta}(\mathbf{X})) - \text{LP}(\hat{I}^{\theta}(\mathbf{X})) \right) \\
&\text{s.t.} \\
& \quad \hat{I}^\theta(\mathbf{X}) \text{ is the solution to the farmer's LP problem in } \eqref{eq:LP}.\\
\end{aligned}
\end{equation}
Here:
\begin{itemize}
    \item $\text{LP}(I^\theta(\mathbf{X})) = \rho_F\left(Y - I^\theta(\mathbf{X}) + \Pi({I^\theta(\mathbf{X})}\right)$ with  $\rho_F$ being the risk measure for the farmer, $I^\theta$ is the indemnity function under the risk-loading factor $\theta$, and $\Pi({I^\theta(\mathbf{X})})$ is formulated by the expectation premium principle $\Pi(I(\mathbf{X}))=(1+\theta)\mathbb{E}(I(\mathbf{X}))$; 
    \item $\gamma>0$ is the penalty factor;
    \item $\mu\geq0$ is the fixed administrative cost factor.
\end{itemize}

This formulation ensures that the penalty term promotes the optimality of the lower-level problem while maintaining the structure of the original bilevel objective. 

The pseudo-code for our penalized bilevel algorithm for Problem 1 is shown below.

\begin{algorithm}[h!]
\caption{Penalty-Based Bilevel Gradient Descent (PBGD) for Problem 1 (feasible set $\mathcal P_{s1}$)}\label{alg1}
\begin{algorithmic}
\State \textbf{Input:} Neural networks $\theta_{nn}$, $\hat{I}^\theta_{nn}$, and $I^\theta_{nn}$; initial learning rate $\alpha_0$, exponential decay factor $\lambda=0.96$; penalty constant $\gamma$; administrative cost factor $\mu$; number of LP objective iterations $N_L$; number of combined objective iterations $N_C$

\For{\textbf{each} iteration $n_C=1, \ldots, N_C$}
      \Comment{Inner decay restarts each combined objective iteration}
    \For{\textbf{each} LP objective iteration $n_L=1, \ldots, N_L$}
        \State Assign the weights from $I^\theta_{nn}$ to $\hat{I}^\theta_{nn}$;
        \State Compute inner step size $\alpha^{(L)}_{n_L} \leftarrow \alpha_0 \cdot \lambda^{n_L} = \alpha_0 \cdot 0.96^{n_L}$;
        \State Update $\hat{I}^\theta_{nn}$ by gradient descent on the LP objective (equation (1)) with step size $\alpha^{(L)}_{n_L}$;
    \EndFor
    \State Compute $\text{LP}(\hat{I}^\theta(\mathbf{X}))$ using $\hat{I}^\theta_{nn}$;
    \State Compute penalty $\gamma \cdot \left( \text{LP}(I^\theta(\mathbf{X})) - \text{LP}(\hat{I}^\theta(\mathbf{X})) \right)$;
    \State Compute combined objective $-\left( \Pi(I^\theta(\mathbf{X})) - (1 + \mu)\mathbb{E}(I^\theta(\mathbf{X})) \right) + \gamma \cdot \left( \text{LP}(I^\theta(\mathbf{X})) - \text{LP}(\hat{I}^\theta(\mathbf{X})) \right)$;
    \State Compute outer step size $\alpha^{(C)}_{n_C} \leftarrow \alpha_0 \cdot \lambda^{n_C} = \alpha_0 \cdot 0.96^{n_C}$;
    \State Update $\theta_{nn}$ and $I^\theta_{nn}$ by gradient descent to minimize the combined objective (equation (14)) with step size $\alpha^{(C)}_{n_C}$;
\EndFor
\State \textbf{Output:} Optimal $\theta$ and $I(\mathbf{X})$.
\end{algorithmic}
\end{algorithm}
The PBGD algorithms for Problems 2 and 3 can be adapted from the above algorithm with some adjustments. For Problem 2, the premium function changes to a distortion premium principle with a power distortion function $g_i(s) = s^{\frac{1}{\rho}}$. This means that the insurer optimizes for two parameters $\theta$ and $\rho$ in this case and the combined objective becomes $-\left( \Pi(I^{\theta, \rho}(\mathbf{X})) - (1 + \mu)\mathbb{E}(I^{\theta, \rho}(\mathbf{X})) \right) + \gamma \cdot \left( \text{LP}(I^{\theta, \rho}(\mathbf{X})) - \text{LP}(\hat{I}^{\theta, \rho}(\mathbf{X})) \right)$. The algorithm then updates for $\theta_{nn}$, $\rho_{nn}$, and $I^\theta_{nn}$. For Problem 3, we use a neural network to model both the insurance payoff function $I^{g}(\mathbf{X})$, and the difference in pricing function $g_i(s)$. The algorithm in this case optimizes for $g_{nn}$ and $I^g_{nn}$.

\ifx
\begin{algorithm}[h!]
\caption{Estimation of the value function $V$}\label{alg2}
\begin{algorithmic}
\State \textbf{Input:} ANNs {\color{red}what is ANN?} for value function $V^\phi$ and policy $\pi^\theta $, number of episodes $N$, transitions $M$, epochs $K$, mini-batch size $B$
\For{each epoch $\kappa=1, \ldots, K$}
    \State Set the gradients of $V^\phi$ to zero;
    \State Sample $B$ states $s_t^{(b)}, b=1, \ldots, B, t \in \mathcal{T}$;
    \State Obtain from $\pi^\theta$ the associated transitions $\left(a_t^{(b, m)}, s_{t+1}^{(b, m)}, c_t^{(b, m)}\right), m=1, \ldots, M$, $t=1,\ldots,T-1$;
    \For{each state $b=1, \ldots, B, t \in \mathcal{T}$}
        \State Compute the predicted values $\hat{v}_t^b=V_t^\phi\left(s_t^{(b)} ; \theta\right)$
        \If{$t=T$}
            \State Set the target value as
            \State $$ 
            v_{T}^b = \max _{\xi^A \in \mathcal{U}\left(\P^\theta\left(\cdot, \cdot\mid s_{T}=s_{T}^{(b)}\right)\right)}\left\{\E_{T, s_{T}^{(b)}}^{\xi^A}\left[W_{T}^{(b, m)}\right]-\rho_{T}^{A,\ast}(\xi^A)\right\}
            $$
        \Else
            \State Set the target value as
            \State $$ \begin{aligned}
                v_t^b = & \max _{\xi^A \in \mathcal{U}\left(\P^\theta\left(\cdot, \cdot \mid s_t=s_t^{(b)}\right)\right)}\left\{\E_{t, s_t^{(b)}}^{\xi^A}\left[c_t^{A,(b, m)}+e^{-r}V_{t+1}^\phi\left(s_{t+1}^{(b, m)}; \theta\right)\right]-\rho_t^{A,\ast}(\xi^A)\right\} p_{x+t} \\
                & +\max _{\xi^B \in \mathcal{U}\left(\P^\theta\left(\cdot, \cdot \mid s_t=s_t^{(b)}\right)\right)}\left\{\E_{t, s_t^{(b)}}^{\xi^B}\left[c_t^{B,(b, m)}\right]-\rho_t^{B,\ast}(\xi^B)\right\} q_{x+t};
            \end{aligned}   
            $$
        \EndIf
    \EndFor
    \State Compute the expected square loss between $v_t^b$ and $\hat{v}_t^b$;
    \State Update $\phi$ by performing an Adam optimizer step;
\EndFor
\State \textbf{Output:} An estimate of the value function $V_t^\phi(s ; \theta) \approx V_t(s ; \theta)$
\end{algorithmic}
\end{algorithm}
\fi

\section{Data}\label{sec:4}
\subsection{Production Loss and Farmer’s Wealth Data}
 We analyze annual county-level soybean data for Iowa from the National Agricultural Statistics Service (NASS), covering 1940-2023. Iowa accounts for a substantial share of U.S. soybean acreage, and the long historical record enables robust analysis of production losses. To achieve stationarity in yields, we detrend county-level soybean yields using a second-order polynomial in time estimated via robust regression and adjust for heteroscedasticity. Following \cite{deng2007there} and \cite{harri2011relaxing}, we normalize historical yields to the 2023 price level so that interannual deviations are comparable over time. This detrending mitigates long-run influences such as climate trends and technological change (e.g., in genetics, management, and equipment), while preserving interannual variability relevant for loss analysis.

Given that the yield data are available only at an annual frequency, the number of historical observations is relatively limited compared with the high-dimensional weather covariates used in this study. To address this limitation, we assume that crop yield losses are homogeneous across time and space. This assumption expands our data set to 3,780 county-years (45 counties $\times$ 84 years). 

In our analysis, we minimize the risk associated with fluctuations in crop yields by working with production losses rather than raw yields. For each observation $n$, we define the loss as the shortfall from the maximum observed yield ($\text{max\_yield} - \text{yield}_n)$), multiplied by a normalized price ($p$):
\begin{equation*}
Y_n = (\text{max\_yield} - \text{yield}_n) \times p,\quad n=1,\ldots,3780,
\end{equation*}
where $p = 1$ denotes one price unit. Monetary quantities are expressed in price-indexed units per acre, with the price index normalized to one. 

The equilibrium solutions satisfy scale invariance, which is guaranteed by Choquet integral representations. This means that scaling all payoffs by any positive factor scales the objectives but leaves the equilibrium allocations and policies unchanged.



\subsection{Climate and Weather Index Data}
The weather data used in this study are obtained from the PRISM Climate Group,\footnote{PRISM is the USDA’s official climatological data set, available at \url{http://prism.oregonstate.edu/}.} which provides detailed monthly meteorological data for the contiguous United States at a 4-km spatial resolution. The data set includes six climate variables: precipitation, maximum and minimum temperatures, maximum and minimum vapor pressure deficits, and dew points, spanning from 1940 to 2023. These variables form a $7 ~\text{indeces} \times 12 ~\text{months}$-dimensional weather index matrix, which is essential for designing an optimal insurance policy. The weather indices used in this study are summarized in Table \ref{tab:weather_indices}.

The relationship between weather indices and crop production losses is often nonlinear and complex due to biological and environmental factors \citep{schlenker2009nonlinear, rigden2020combined}. Additionally, interactions between weather indices can compound their impact on yield losses. For example, while minimum vapor pressure deficit in October and precipitation in October may not individually influence production losses, their combined effect can create a substantial and nonlinear impact. These complexities highlight the limitations of linear models, which are commonly used in existing index insurance contracts. 
This underscores the need for more sophisticated models in insurance design.

\begin{table}[htbp]
\centering
\caption{Weather Indices}
\begin{tabular}{p{2.5cm}p{12cm}} 
\hline
Variable & Description \\
\hline
pcpnk & Total precipitation (rain and melted snow) for month $k$ (mm) \\
tmaxk & Daily maximum temperature averaged over all days in month $k$ (\textdegree C) \\
tmink & Daily minimum temperature averaged over all days in month $k$ (\textdegree C) \\
dptk & Daily mean dew point temperature averaged over all days in month $k$ (\textdegree C) \\
vpdmaxk & Daily maximum vapor pressure deficit averaged over all days in month $k$ (hPa) \\
vpdmink & Daily minimum vapor pressure deficit averaged over all days in month $k$ (hPa) \\
$k$ & Calendar month, $k=1,2,\ldots,12$ for January-December \\
\hline
\end{tabular}
\begin{flushleft}
\textit{Notes.} This table summarizes the weather variables available from the PRISM data set. The sample period is from 1940 to 2023.
\end{flushleft}
\label{tab:weather_indices}
\end{table}

\section{Numerical Results}\label{sec:5}
The numerical results presented in this section are structured as follows: Section \ref{sec:results_nn} shows the training and equilibrium results for insurance payoffs, $I(X)$, modeled using fully connected NNs. Section \ref{sec:results_sa} presents a sensitivity analysis using fully connected neural networks, varying the parameters $\lambda$ and $\rho$ in the farmer's distortion functions $g_f(s) = \lambda s + (1-\lambda)\min\left\{\frac{s}{1-\alpha}, 1\right\}$ and $g_f(s) = s^{\frac{1}{\rho}}$, respectively. The sensitivity analysis is conducted for both index and indemnity insurance. Section \ref{sec:results_cnn} presents training and equilibrium results when CNNs are used to model the insurance payoff function $I(X)$. A comparison of the resulting insurance payoff functions $I(X)$ is also presented in this section. Section \ref{sec:results_two_para} extends the insurer's pricing strategy to a two-parameter case, in which the insurer chooses both $\theta$, the risk-loading factor, and $\rho$ in its premium distortion function $g_i(s) = s^{\frac{1}{\rho}}$. Section \ref{sec:results_pk} further extends the analysis to a general pricing strategy for the insurer, where neural networks are used to model both the insurance payoff $I(X)$ and the insurer's pricing function $g_i$. 

\subsection{Problem 1 ($\mathcal{P}_{s1}$): One-parameter Premium Model}
\subsubsection{Fully Connected Neural Networks}\label{sec:results_nn}
This section presents the numerical solutions to the bilevel optimization Problem 1 when fully connected NNs are used to model the insurance payoff function. The model validation results used to select the optimal model are shown in Table \ref{tab:validation_nn}. 

\begin{table}[ht]
\centering
\caption{Model Validation of Neural Networks ($\lambda = 0$).}
\footnotesize
\begin{tabular}{ccccccccc}
\hline
\multirow{2}{*}{\textbf{Single hidden layer}} & \multicolumn{2}{c}{\textbf{{[}8{]}}} & \multicolumn{2}{c}{{[}16{]}} & \multicolumn{2}{c}{{[}32{]}} & \multicolumn{2}{c}{{[}64{]}} \\ \cline{2-9} 
 & \textbf{Train} & \textbf{Validation} & Train & Validation & Train & Validation & Train & Validation \\ \hline
UP Loss & \textbf{-3.1783} & \textbf{-3.1485} & -0.0005 & -0.0012 & -0.0597 & -0.0495 & 0.0000 & 0.0000 \\ \hline
Penalized Loss & \textbf{-2.9113} & \textbf{-5.2283} & 0.0056 & -0.0142 & -0.0936 & 0.0279 & 0.0000 & 0.0000 \\ \hline
\multirow{2}{*}{\textbf{Multiple hidden layers}} & \multicolumn{2}{c}{\textbf{{[}8 - 8{]}}} & \multicolumn{2}{c}{{[}8 - 8 - 8{]}} & \multicolumn{2}{c}{{[}8 - 8 - 8 - 8{]}} & \multicolumn{2}{c}{{[}8 - 8 - 8 - 8 - 8{]}} \\ \cline{2-9} 
 & \textbf{Train} & \textbf{Validation} & Train & Validation & Train & Validation & Train & Validation \\ \hline
UP Loss & \textbf{-3.5495} & \textbf{-3.6050} & 0.0000 & 0.0000 & 0.0000 & 0.0000 & 0.0000 & 0.0000 \\ \hline
Penalized Loss & \textbf{-3.6556} & \textbf{-3.6752} & 0.0000 & 0.0000 & 0.0000 & 0.0000 & 0.0000 & 0.0000 \\ \hline
\end{tabular}
\label{tab:validation_nn}
\end{table}

Figure \ref{fig:train_combo} shows the convergence of the UP and LP loss curves for index insurance, where the farmer's risk measure is $g_f(s) = 0.5s + 0.5\min\left\{\frac{s}{1-\alpha}, 1\right\}~\text{with}~\alpha = 0.8$. The upper-problem loss curve stabilizes around -1.6068, which indicates an equilibrium profit of 1.6068 for the insurer. The lower-problem loss stabilizes around 42.4050, which indicates that the farmer's risk measure at equilibrium is 42.4050. The weighting parameter $\lambda$ between CVaR and the expected risk adjusts the risk appetite of the farmer. A higher $\lambda$ leads to less weight assigned to the CVaR component, which means that the farmer is less risk averse than if she had CVaR as its risk measure. The optimal insurance payoff for this case with the convex-combination risk measure is shown in Figure \ref{fig:index_ix_combo}. In equilibrium, the optimal risk-loading factor $\theta^*$ for the insurer is 0.4130, which leads to a profit of 1.6086. The optimal insurance payoff function exhibits a stop-loss pattern. We make a comparison here with indemnity insurance, which is based directly on the loss $Y$. For indemnity insurance under the same risk measure for the farmer (convex-combination risk measure), the optimal risk-loading factor for the insurer is $\theta^* = 0.2812$, which corresponds to a profit of 2.3633. The payoff function in this case of indemnity insurance is stop-loss, as shown in Figure \ref{fig:indemnity_ix_combo}. 

\begin{figure}[htbp]
    \centering
    \begin{subfigure}[b]{0.45\textwidth}
        \centering
        \includegraphics[width=\textwidth]{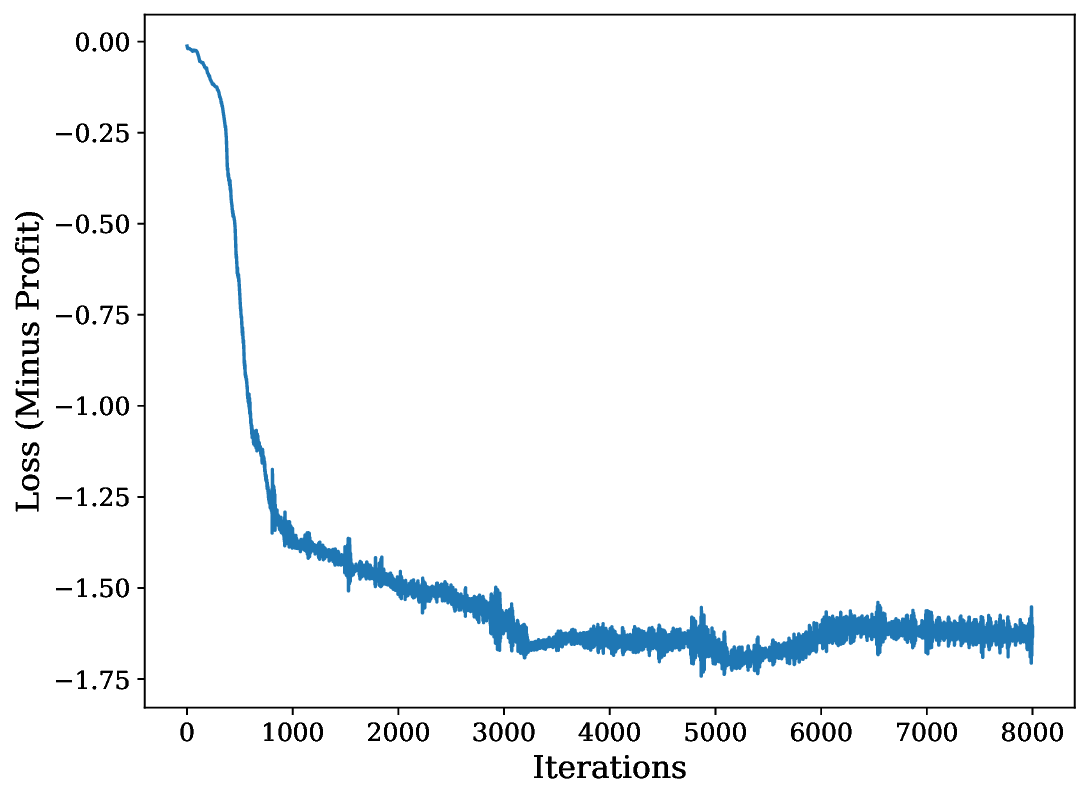} 
        \caption{Upper-problem loss}
        \label{fig:combo_up}
    \end{subfigure}
    \hfill
    \begin{subfigure}[b]{0.45\textwidth}
        \centering
        \includegraphics[width=\textwidth]{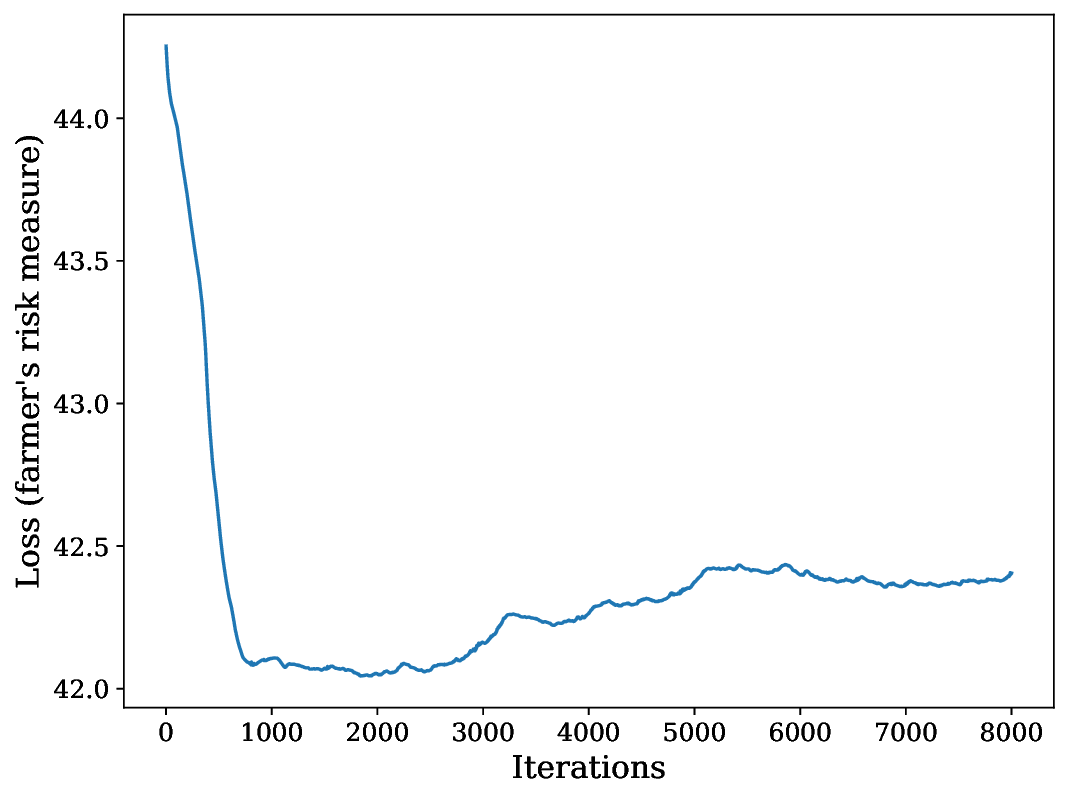} 
        \caption{Lower-problem loss}
        \label{fig:combo_lp}
    \end{subfigure}
    \caption{Training curves of using NN to model insurance payoff under the farmer's risk distortion function $g_f(s) = 0.5s + 0.5\min\left\{\frac{s}{1-\alpha}, 1\right\} ~\text{with}~ \alpha = 0.8$. Figure (a) shows the evolution of the upper-problem loss (negative insurer's profit) across iterations. Figure (b) depicts the evolution of the lower-problem loss (farmer's risk measure) across iterations.}
    \label{fig:train_combo}
\end{figure}

In addition, we change the farmer's risk measure to pure CVaR ($\lambda = 0$ in the convex-combination risk measure) and investigate the equilibrium solution in this sequential game with a more risk-averse farmer than in the previous case. The insurer's optimal risk loading in the equilibrium solution is $\theta^* = 0.5500$, which corresponds to a profit of 3.5495. Note that the risk loading $\theta$ and profit are higher than in the previous case. The optimal payoff function, which shows a stop-loss pattern, is shown in Figure \ref{fig:index_ix_cvar}. \cite{cai2008optimal} show that stop-loss insurance is optimal under the CVaR risk measure, which supports our study. We also compare with indemnity insurance under the same risk measure. The optimal risk loading is $\theta^* = 0.4853$, which corresponds to a profit of 4.8277. The risk loading $\theta^*$ is lower and the profit is higher for indemnity insurance than for index insurance, which could be caused by the existence of basis risk in index insurance. This outcome is also observed in the previous case of convex-combination risk measure. Moreover, the optimal insurance payoff functions for both types of insurance contracts have similar shapes, but the insurance payoff functions for index insurance have more noise in the payoff patterns compared to those for indemnity insurance. The noise in the index insurance payoff functions is also observed in \citep{chen2024managing}.

\begin{figure}[H]
    \centering
    \begin{subfigure}[b]{0.45\textwidth} 
        \centering
        \includegraphics[width=\textwidth]{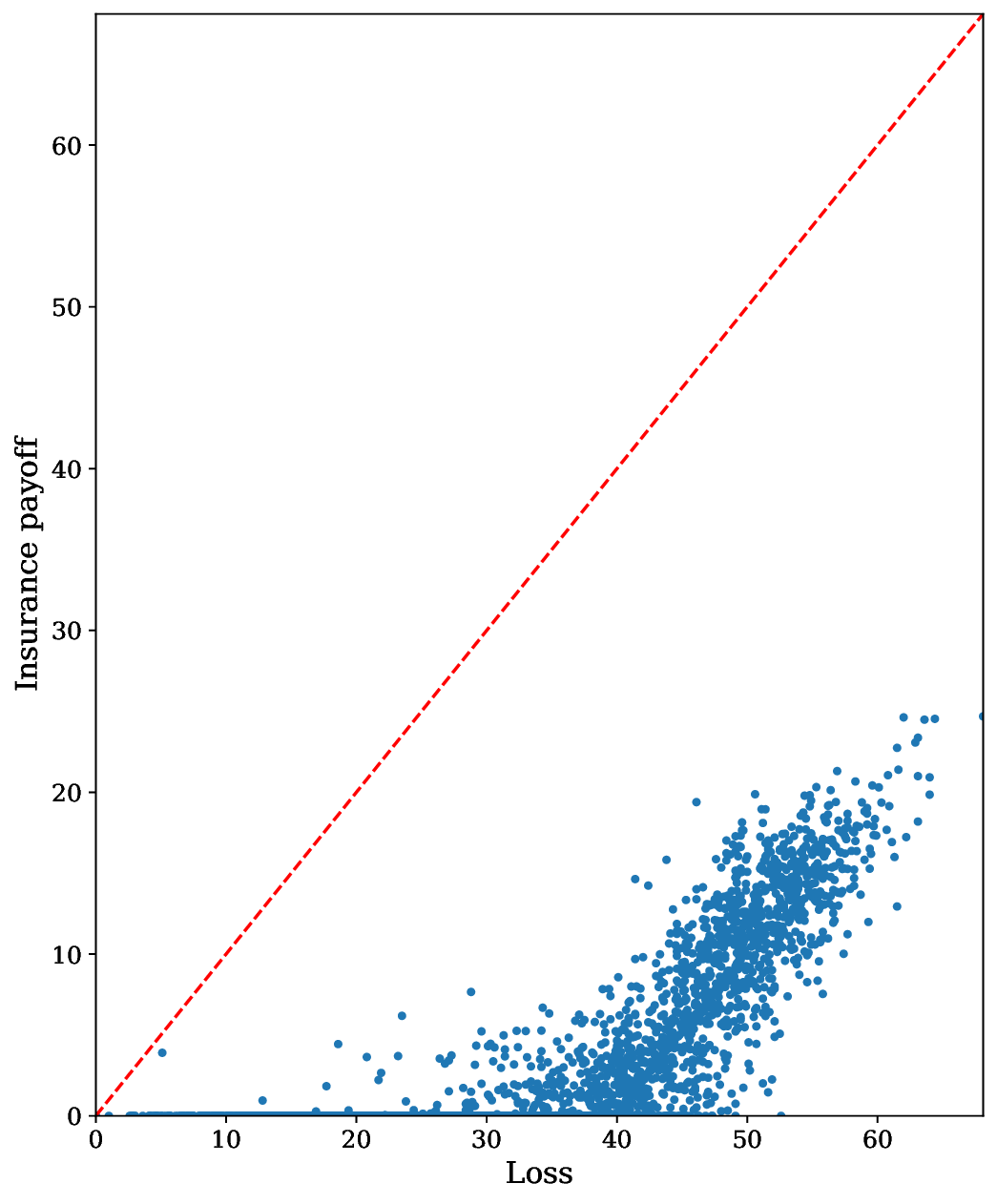}
        \caption{}
        \label{fig:index_ix_combo}
    \end{subfigure}
    \hfill
    \begin{subfigure}[b]{0.45\textwidth} 
        \centering
        \includegraphics[width=\textwidth]{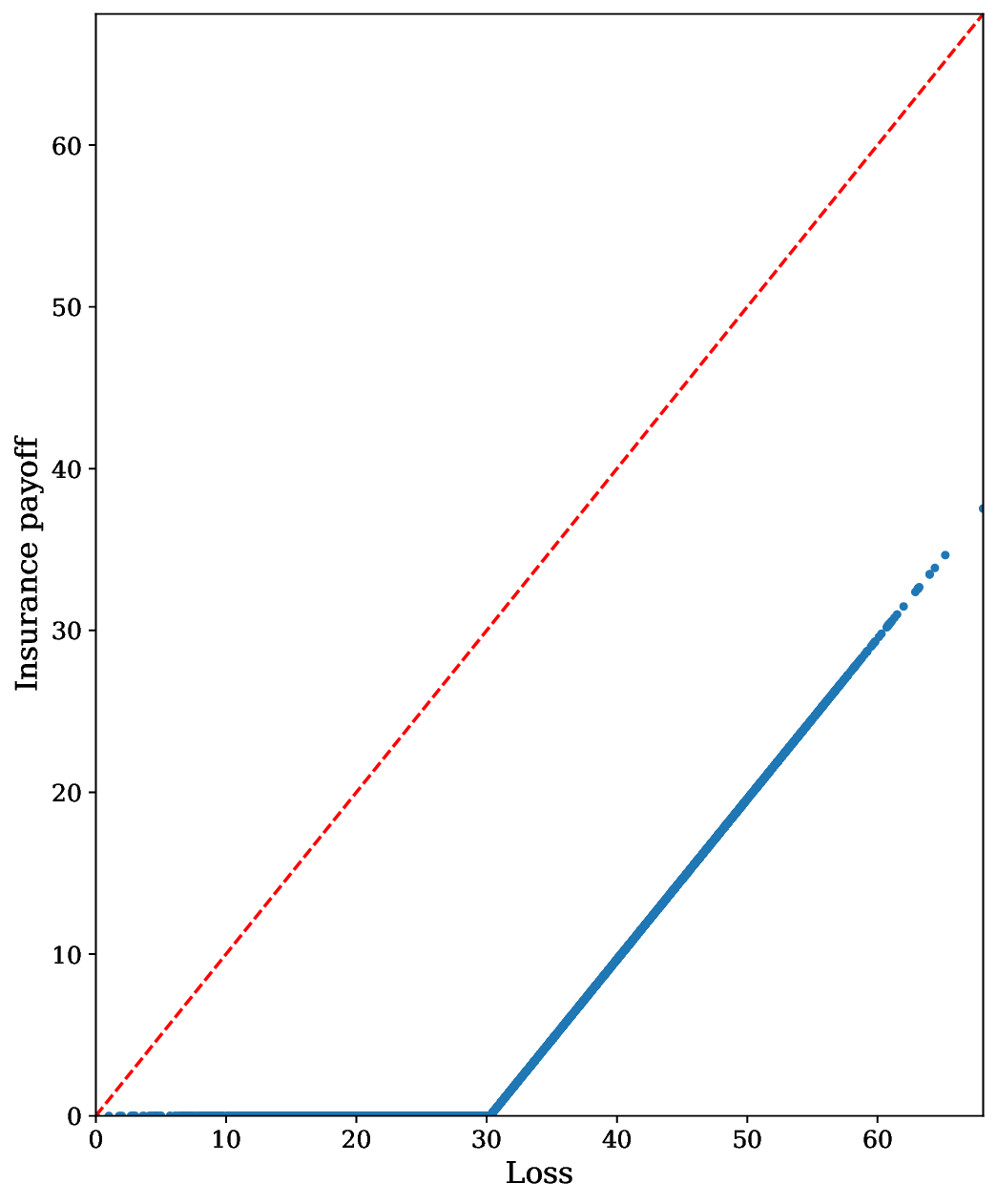} 
        \caption{}
        \label{fig:indemnity_ix_combo}
    \end{subfigure}
    
    \vspace{0.5cm}
    
    \begin{subfigure}[b]{0.45\textwidth} 
        \centering
        \includegraphics[width=\textwidth]{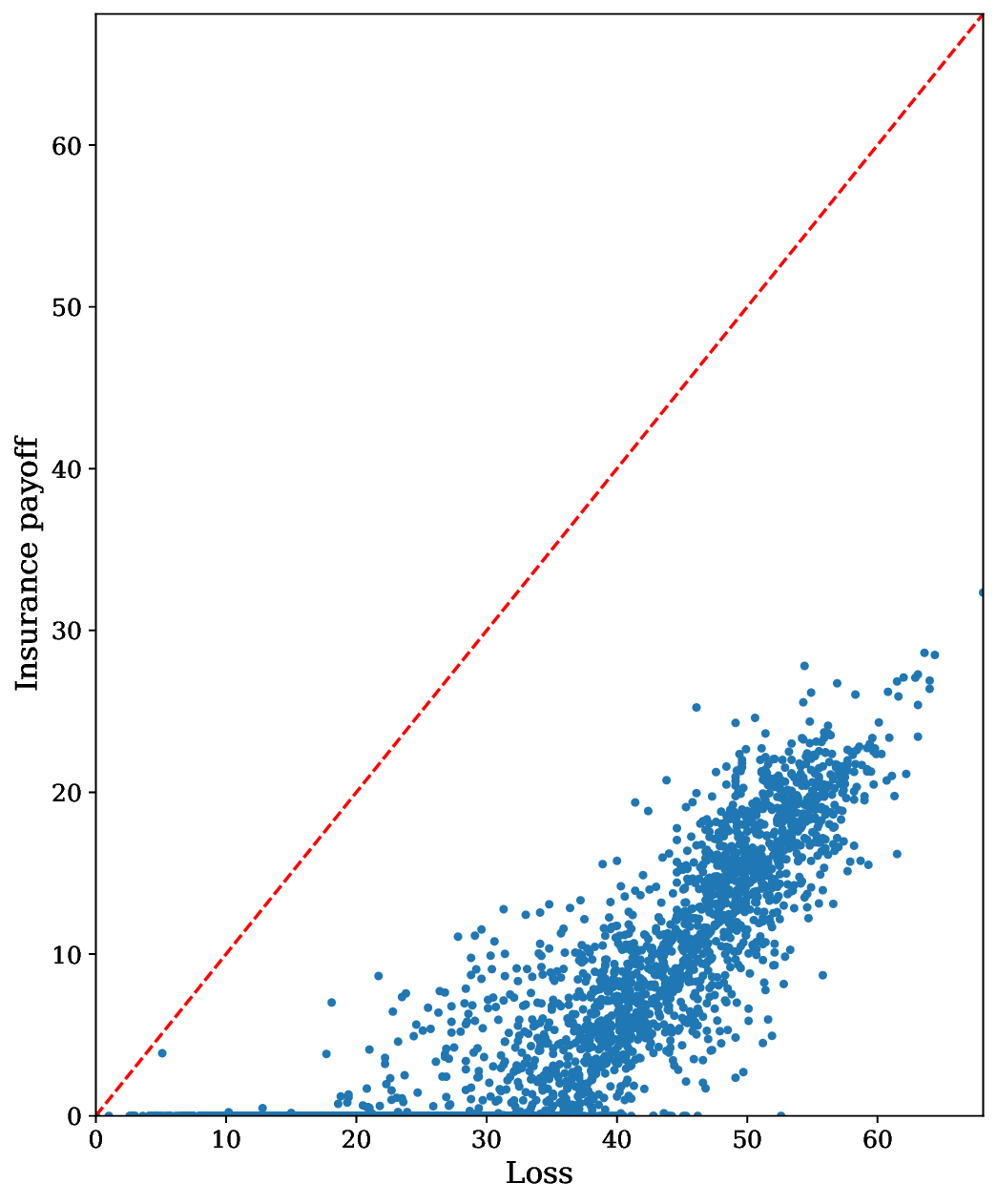}
        \caption{}
        \label{fig:index_ix_cvar}
    \end{subfigure}
    \hfill
    \begin{subfigure}[b]{0.45\textwidth} 
        \centering
        \includegraphics[width=\textwidth]{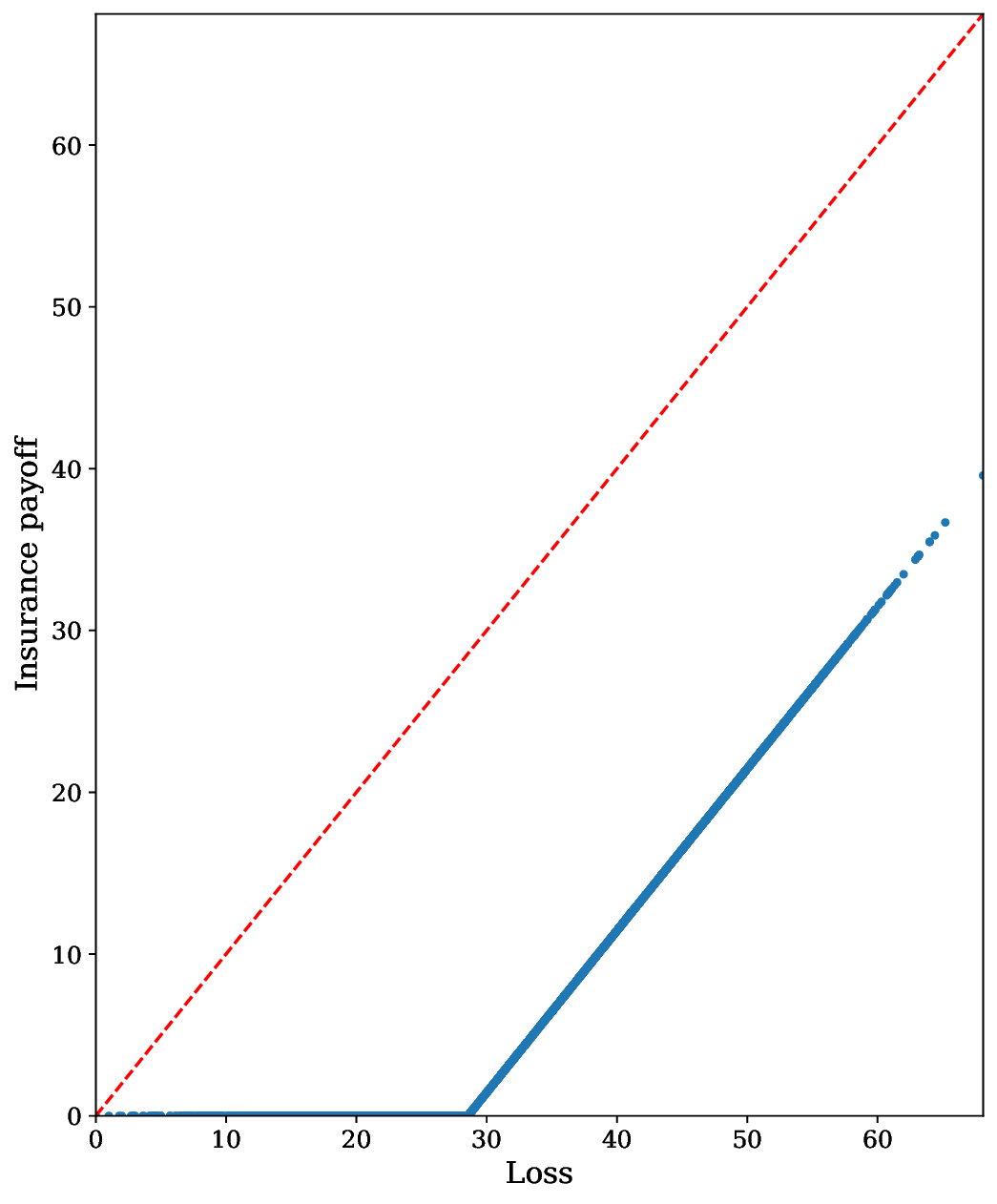}
        \caption{}
        \label{fig:indemnity_ix_cvar}
    \end{subfigure}
    
    \caption{Comparison of equilibrium solutions under risk measure CVaR and convex combination for Index and indemnity insurance. Figures (a) and (b) show the optimal insurance payoff functions for index and indemnity insurance respectively. The farmer's risk distortion function adopted in these two figures is the convex combination of CVaR and expected risk $g_f(s) = 0.5s + 0.5\min\left\{\frac{s}{1-\alpha}, 1\right\}$ with $\alpha = 0.8$. Figures (c) and (d) illustrate a different scenario where the farmer's risk distortion risk measure is CVaR with $\alpha = 0.8$.} 
    \label{fig:index_indemnity_cvar_combo}
\end{figure}

\subsubsection{Sensitivity Analysis} \label{sec:results_sa}
In the first sensitivity analysis, we vary the weighting parameter in the distortion function $g_f(s) = \lambda s + (1 - \lambda) \cdot \min\{\frac{s}{1-\alpha}, 1\}$. The comparison of insurance coverage for $\alpha = 0.95$ and $\lambda = 0.1, 0.5, 0.9$ is shown in Figure \ref{fig:sa_lambda_index}. We observe that a higher $\lambda$ leads to higher deductibles. The same pattern is consistent for the indemnity insurance shown in Figure \ref{fig:sa_lambda_indemnity}.

The second sensitivity analysis involves varying the parameter $\rho$ in the distortion function $g_i(s) = s^{\frac{1}{\rho}}$ for the premium. We compare the optimal risk-loading factor $\theta$ for $\alpha = 0$ and $\rho = 1.0, 1.5, 2$. As shown in Figure \ref{fig:sa_rho_index}, a larger $\rho$ leads to slightly higher deductibles. The same pattern is also observed for the indemnity insurance in Figure \ref{fig:sa_rho_indemnity}. The summary of the optimal risk-loading factor $\theta^*$ found in the two sensitivity studies is shown in Table \ref{tab:sa_summary}. We do not allow a negative loading factor $\theta$ in the final solution, as this leads to an arbitrage opportunity for the farmer. If $\theta$ would be negative, an arbitrage opportunity would exist since we do not restrict the coverage to be strictly smaller than or equal to the actual loss.

\begin{figure}[htbp]
    \centering
    \begin{subfigure}[b]{0.45\textwidth}
        \centering
        \includegraphics[width=\textwidth]{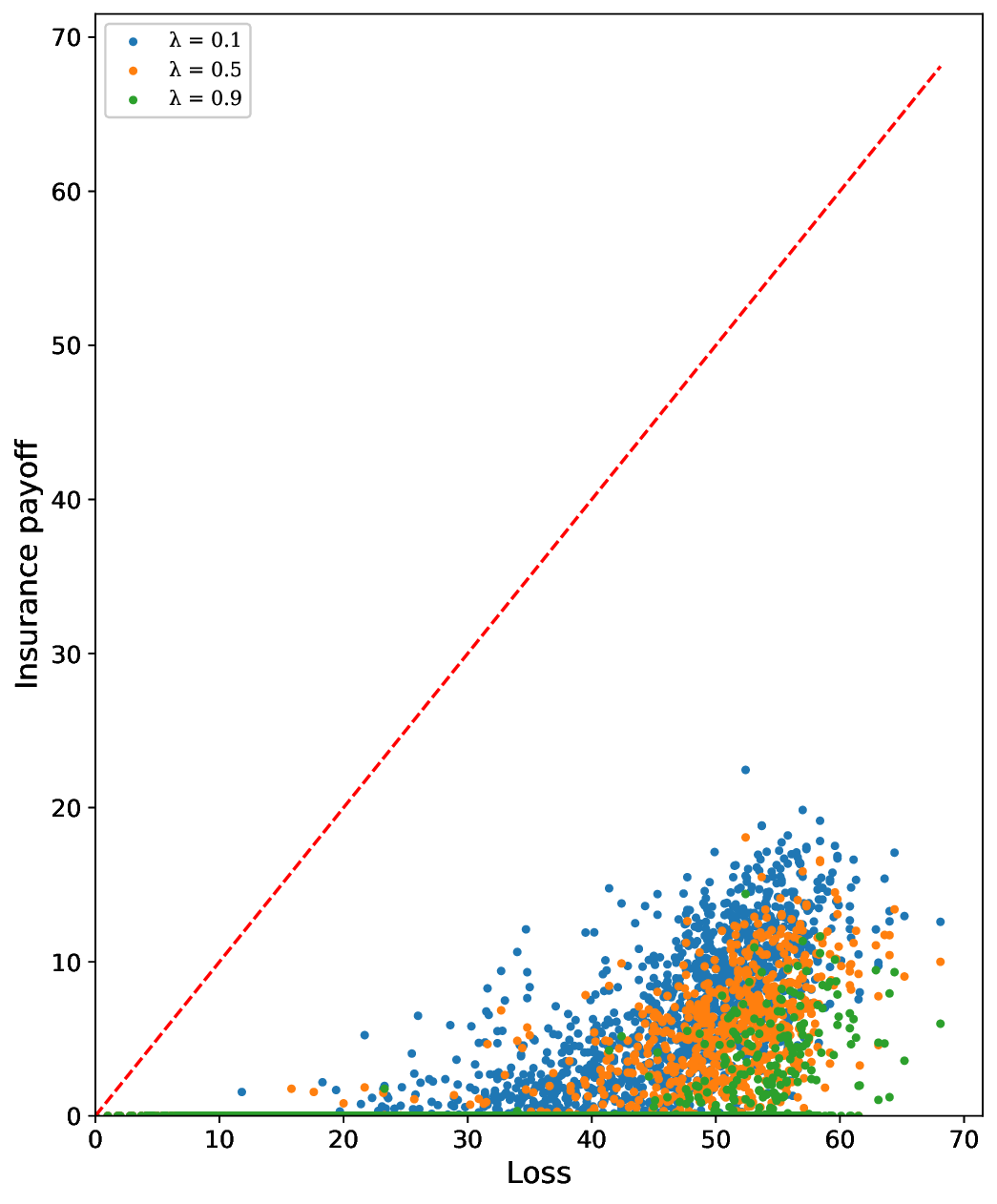}
        \caption{}
        \label{fig:sa_lambda_index}
    \end{subfigure}
    \hfill
    \begin{subfigure}[b]{0.45\textwidth}
        \centering
        \includegraphics[width=\textwidth]{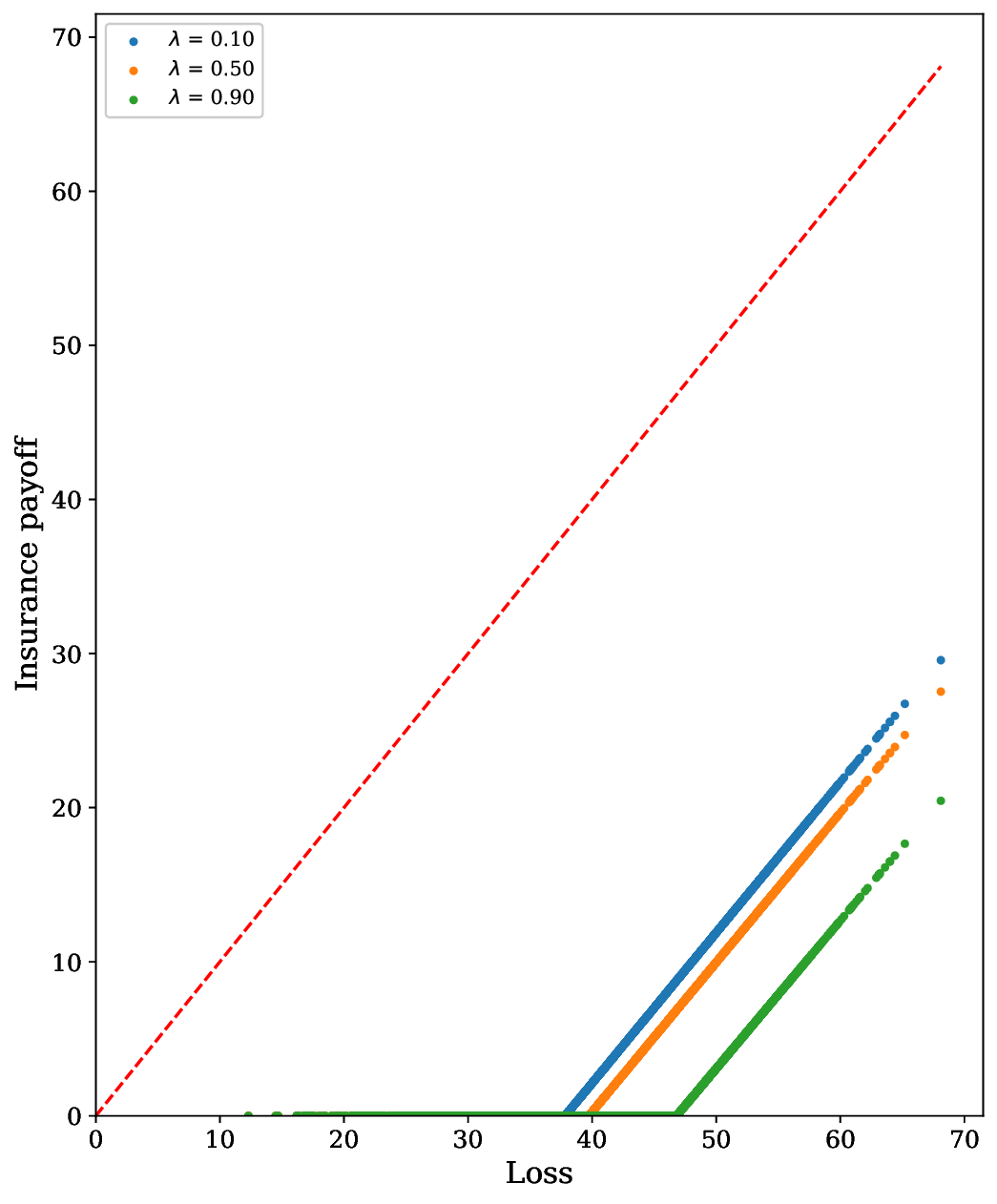} 
        \caption{}
        \label{fig:sa_lambda_indemnity}
    \end{subfigure}
    
    \vspace{10pt}
    
    \begin{subfigure}[b]{0.45\textwidth}
        \centering
        \includegraphics[width=\textwidth]{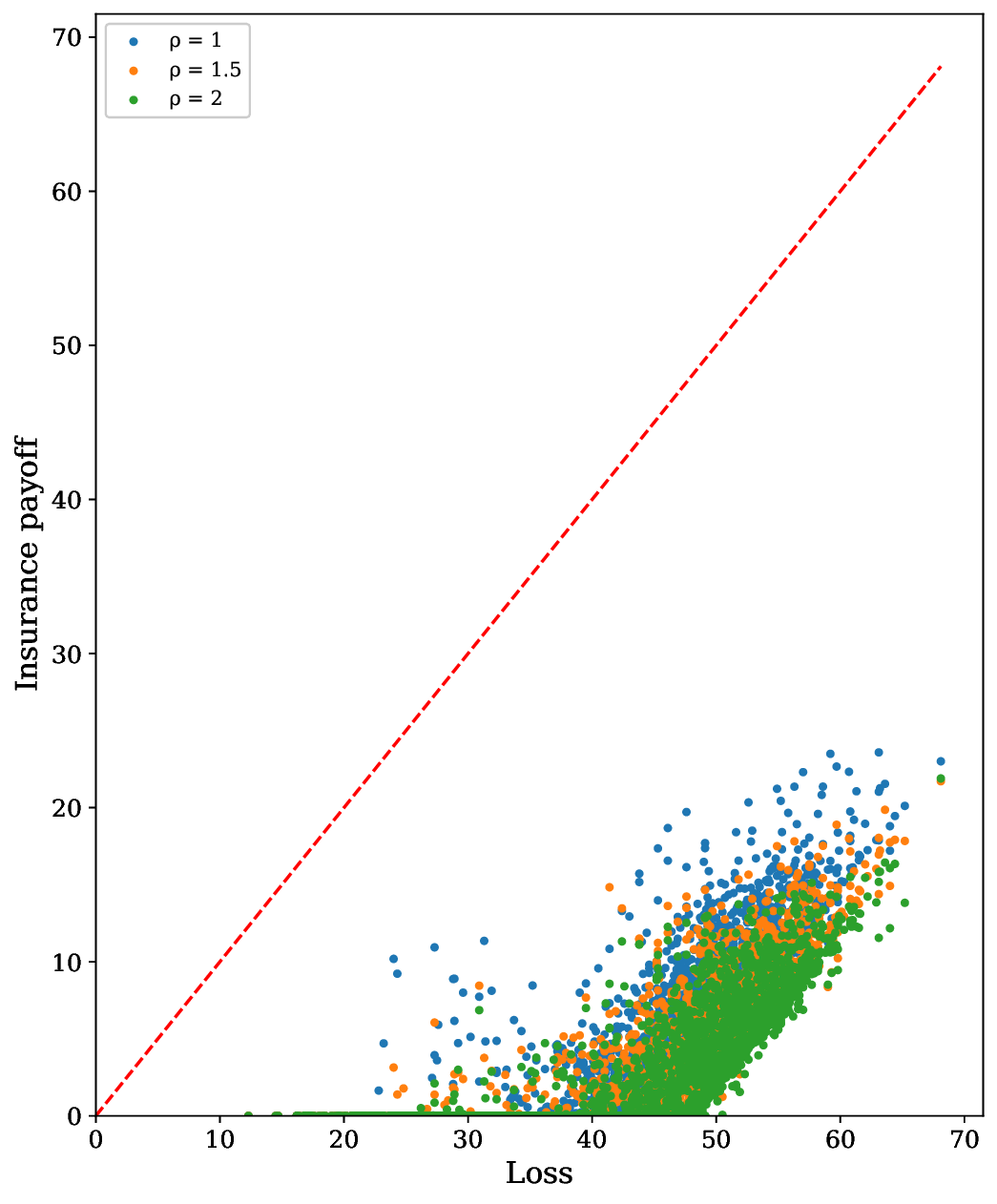}
        \caption{}
        \label{fig:sa_rho_index}
    \end{subfigure}
    \hfill
    \begin{subfigure}[b]{0.45\textwidth}
        \centering
        \includegraphics[width=\textwidth]{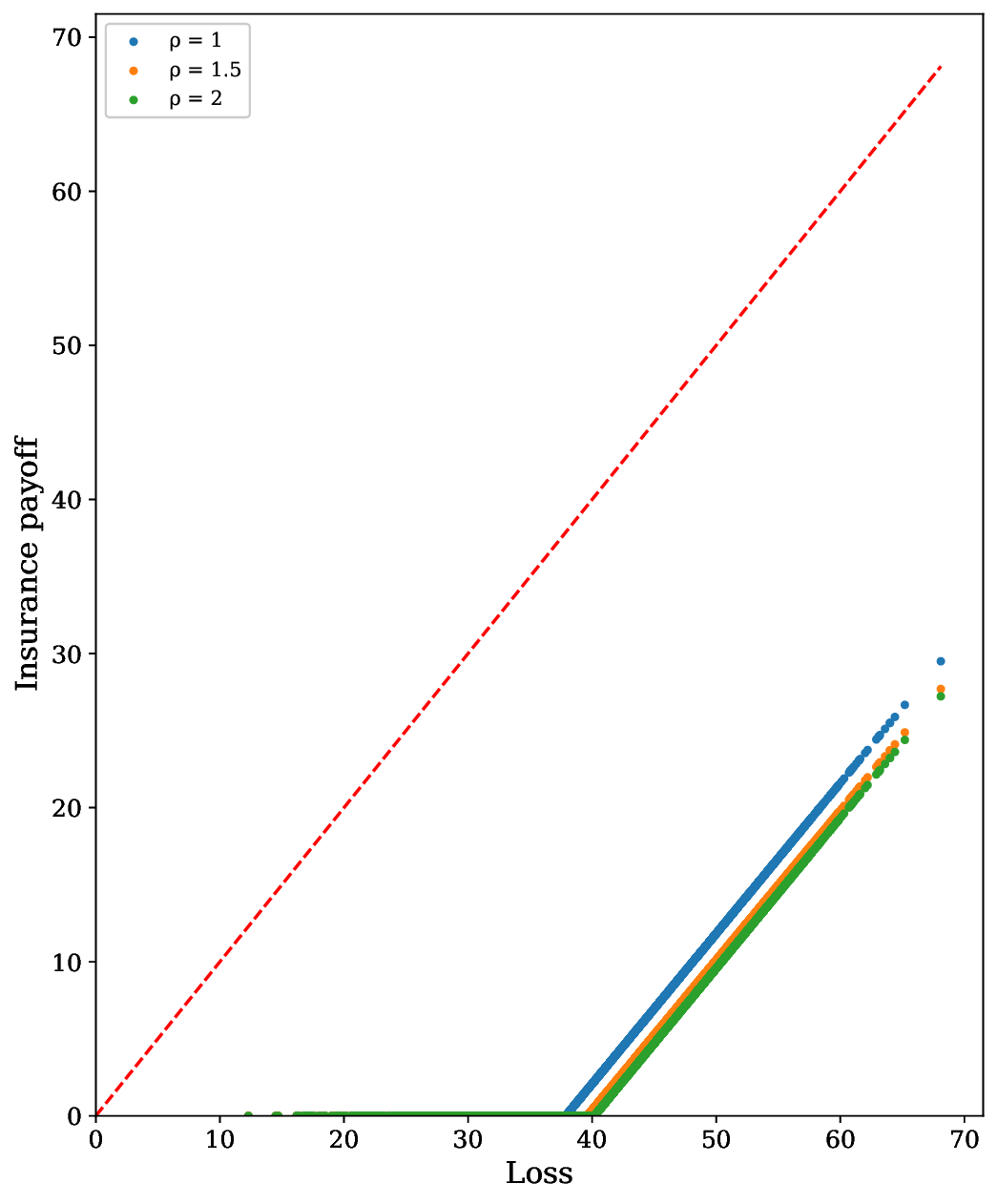}
        \caption{}
        \label{fig:sa_rho_indemnity}
    \end{subfigure}
    
    \caption{Sensitivity analysis based on $\lambda$ and $\rho$. The impact of $\lambda$ on the equilibrium solution for index and indemnity insurance are shown in Figure (a) and (b) respectively. Figure (a) demonstrates the difference in optimal insurance payoff $I^*(\mathbf{X})$ when the parameter $\lambda$ changes between 0.1, 0.5 and 0.9 in the farmer's risk distortion function $g_f(s) = \lambda s + (1 - \lambda) \cdot \min\{\frac{s}{1-\alpha}, 1\}$. Figure (b) shows the same impact of changing $\lambda$ on the optimal insurance coverage but on indemnity insurance, where the insurance contract is directly written on the loss $Y$.
    The impact of $\rho$ on insurance payoff function for index and indemnity insurance under insurer's premium distortion function $g_i(s) = s^{\frac{1}{\rho}}$ are shown in Figure (c) and (d), respectively. Figure (c) shows the difference in insurance payoff when $\rho$ changes between 1, 1.5 and 2. Figure (d) shows similar differences in the insurance payoff but for indemnity insurance.}
    \label{fig:sa_lambda_rho_index_indemnity}
\end{figure}

\begin{table}[htbp]
\centering
\caption{Summary of the two sensitivity analyses.}
\begin{tabular}{@{}cccc|ccc@{}}
\hline
 & \multicolumn{3}{c}{$I(\mathbf{X})$} & \multicolumn{3}{c}{$I(\mathbf{Y})$} \\ \hline
$\lambda$ & 0.1 & 0.5 & 0.9 & 0.1 & 0.5 & 0.9 \\
$\theta^*$ & 0.7542 & 0.5578 & 0.3068 & 0.6451 & 0.3724 & 0.2591 \\ \hline
$\rho$ & 1.0 & 1.5 & 2 & 1.0 & 1.5 & 2 \\
$\theta^*$ & 1.0181 & 1.0333 & 1.0802 & 0.8796 & 0.9382 & 0.9855 \\ \hline
\end{tabular}
\label{tab:sa_summary}
\end{table}

\subsubsection{Convolutional Neural Networks (CNNs)}\label{sec:results_cnn}
This subsection investigates numerical solutions to the same optimization problem (Problem 1) as in the previous subsection, using CNNs to model the insurance payoff function.

Figure \ref{fig:train_combo_cnn} shows the convergence of the loss curves when CNNs are used to model the insurance payoff function, where the farmer’s risk measure is the convex combination of CVaR and expected risk: $g_f(s) = 0.5 s + 0.5\min\left\{\frac{s}{1-\alpha}, 1\right\} ~\text{with}~ \alpha = 0.8$. The upper-problem loss curve stabilizes around -2.2553, indicating an equilibrium profit of 2.2553 for the insurer. The lower-problem loss curve stabilizes around 40.4945, indicating the farmer’s risk measure at equilibrium. The model validation results for CNN models are shown in Table \ref{tab:validation_cnn}. 

\begin{figure}[h!]
    \centering
    \begin{subfigure}[b]{0.45\textwidth}
        \centering
        \includegraphics[width=\textwidth]{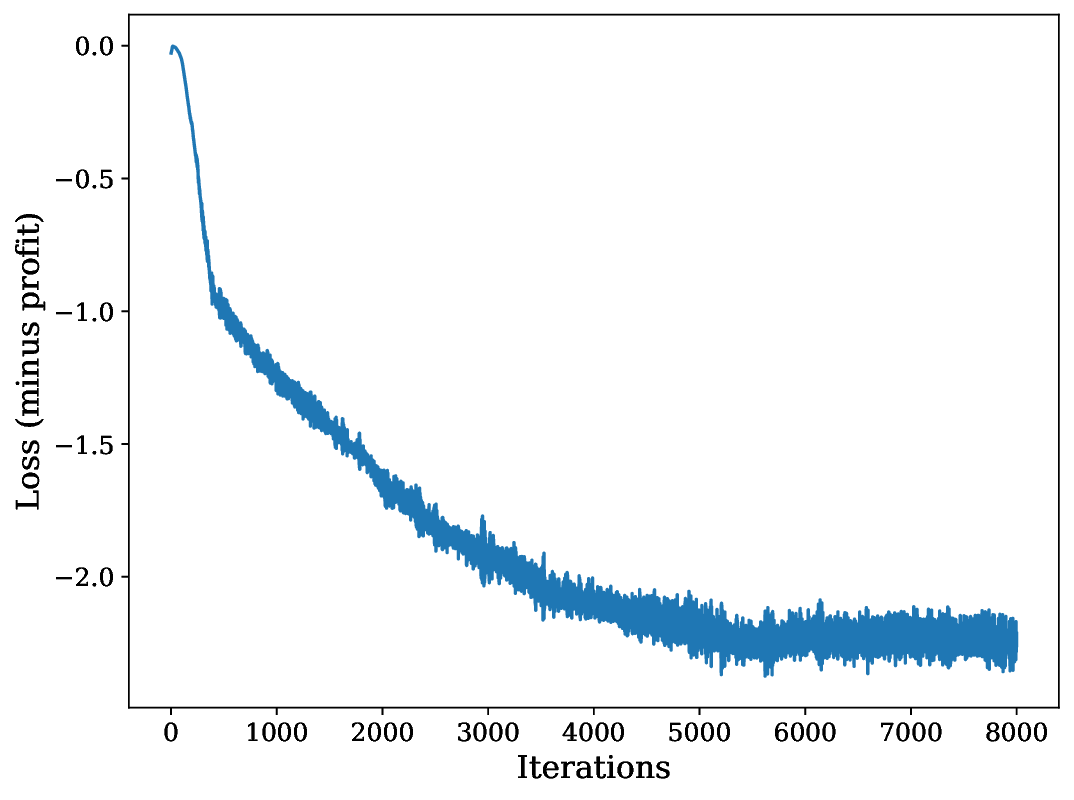}
        \caption{UP Loss}
        \label{fig:combo_up_cnn}
    \end{subfigure}
    \hfill
    \begin{subfigure}[b]{0.45\textwidth}
        \centering
        \includegraphics[width=\textwidth]{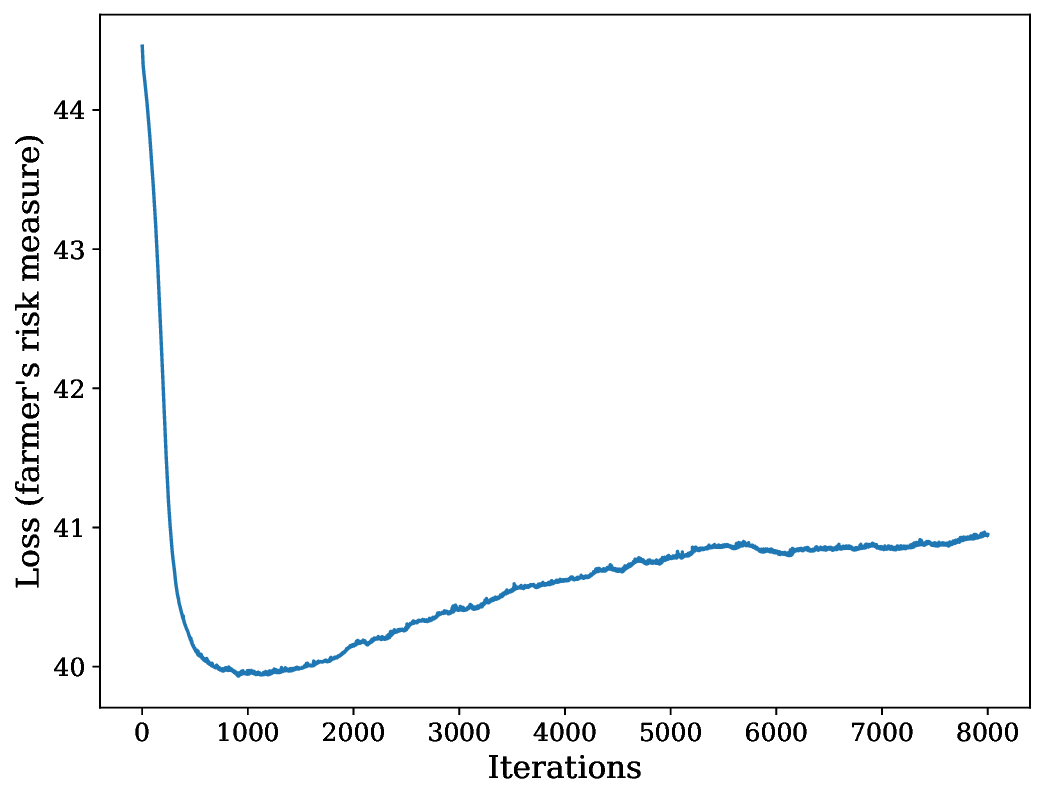} 
        \caption{LP loss}
        \label{fig:combo_lp_cnn}
    \end{subfigure}
    \caption{Training curves for using CNNs to model the insurance payoff under the farmer’s risk distortion function $g_f(s) = 0.5s + 0.5\min\left\{\frac{s}{1-\alpha}, 1\right\} ~\text{with}~ \alpha = 0.8$. Figure (a) shows the evolution of the upper-problem loss (the negative of the insurer’s profit) across iterations. Figure (b) shows the evolution of the lower-problem loss (the farmer’s risk measure) across iterations.}
    \label{fig:train_combo_cnn}
\end{figure}

\begin{table}[h!]
\centering
\caption{Model Validation for Convolutional Neural Networks ($\lambda = 0$).}
\footnotesize
\begin{tabular}{ccccccccc}
\hline
\multirow{2}{*}{\textbf{\begin{tabular}[c]{@{}c@{}}Convolutional \\ Layers\end{tabular}}} & \multicolumn{2}{c}{{[}16 - 16{]}} & \multicolumn{2}{c}{\textbf{{[}32 - 32{]}}} & \multicolumn{2}{c}{{[}16 - 16 - 16{]}} & \multicolumn{2}{c}{{[}32 - 32- 32{]}} \\ \cline{2-9} 
 & Training & Validation & \textbf{Training} & \textbf{Validation} & Training & Validation & Training & Validation \\ \hline
UP Loss & 0.0000 & 0.0000 & \textbf{-4.0703} & \textbf{-4.0801} & 0.0000 & 0.0000 & 0.2473 & 0.2377 \\ \hline
Penalized Loss & 0.0000 & 0.0000 & \textbf{-4.2847} & \textbf{-3.7446} & 0.0000 & 0.0000 & 2.8403 & 0.8297 \\ \hline
\multirow{2}{*}{\textbf{\begin{tabular}[c]{@{}c@{}}Fully connected layers\\ (single hidden layer)\end{tabular}}} & \multicolumn{2}{c}{{[}8{]}} & \multicolumn{2}{c}{{[}16{]}} & \multicolumn{2}{c}{\textbf{{[}32{]}}} & \multicolumn{2}{c}{{[}64{]}} \\ \cline{2-9} 
 & Training & Validation & Training & Validation & \textbf{Training} & \textbf{Validation} & Training & Validation \\ \hline
UP Loss & -4.2222 & -4.3253 & -4.0660 & -4.1693 & \textbf{-4.0703} & \textbf{-4.0801} & -3.6794 & -3.7555 \\ \hline
Penalized Loss & -11.2273 & -5.5120 & -3.6540 & -5.2039 & \textbf{-4.2847} & \textbf{-3.7446} & -5.1027 & -4.2090 \\ \hline
\multirow{2}{*}{\textbf{\begin{tabular}[c]{@{}c@{}}Fully connected layers\\ (multiple hidden layers)\end{tabular}}} & \multicolumn{2}{c}{{[}32 - 32{]}} & \multicolumn{2}{c}{{[}32 - 32 - 32{]}} & \multicolumn{2}{c}{{[}32 - 32 - 32 - 32{]}} & \multicolumn{2}{c}{{[}32 - 32 - 32 - 32 - 32{]}} \\ \cline{2-9} 
 & Training & Validation & Training & Validation & Training & Validation & Training & Validation \\ \hline
UP Loss & 0.0000 & 0.0000 & -4.0317 & -4.0989 & 0.0000 & 0.0000 & 0.0000 & 0.0000 \\
Penalized Loss & 0.0000 & 0.0000 & -6.0228 & -2.6661 & 0.0000 & 0.0000 & 0.0000 & 0.0000 \\ \hline
\end{tabular}
\label{tab:validation_cnn}
\end{table}

Figure \ref{fig:cnn_ix_cvar_combo} compares the insurance payoff functions $I(\mathbf{X})$ modeled using CNNs under two different risk measures (CVaR and convex combination). Figure \ref{fig:index_ix_combo_cnn} shows the case in which the  convex combination of CVaR and expected risk is used as the farmer's risk measure. The optimal risk-loading factor $\theta^*$ and profit for the insurer in this case are 0.3182 and 2.2553, respectively. Figure \ref{fig:index_ix_cvar_cnn} shows the scenario in which the farmer’s risk measure is set to CVaR; the optimal risk-loading factor and the insurer’s profit are 0.3961 and 4.0703, respectively. Note that the insurance payoff function modeled using CNNs contains less noise in the payoff patterns than those modeled using NNs, which demonstrates the ability of CNNs to reduce noise and produce more robust models. In addition, the profits obtained using CNNs to model the insurance payoff function are closer to those of indemnity insurance. This outcome holds under both risk measures considered in this study. This is due to the fact that the insurance payoff function modeled using CNNs contains less noise and thus less basis risk. This highlights the potential of CNNs in reducing the basis risk in index insurance pricing strategies and achieving an optimal solution closer to that of indemnity insurance. The differences in the optimal insurance payoff functions can be seen by comparing Figure \ref{fig:index_indemnity_cvar_combo} for the NNs with Figure \ref{fig:cnn_ix_cvar_combo}. A summary of the equilibrium solutions under the two risk measures and the two modeling strategies for the insurance payoff function is shown in Table \ref{tab:summary_equilibrium}.

\begin{table}[htbp]
\caption{Summary of Equilibrium Solutions for Bilevel Optimization Problem 1.}
\centering
\begin{tabular}{@{}cccccc@{}}
\toprule
\multirow{2}{*}{\textbf{Models}} & \multirow{2}{*}{\textbf{Risk Measure}} & \multicolumn{2}{c}{\textbf{$\theta^*$}} & \multicolumn{2}{c}{\textbf{Profit}} \\ \cmidrule(l){3-6} 
 &  & \multicolumn{1}{l}{Index} & \multicolumn{1}{l}{Indemnity} & \multicolumn{1}{l}{Index} & \multicolumn{1}{l}{Indemnity} \\ \midrule
\multirow{2}{*}{NN} & Combination & 0.4130 & 0.2812 & 1.6086 & 2.3633 \\ \cmidrule(l){2-6} 
 & CVaR & 0.5500 & 0.4853 & 3.5495 & 4.8277 \\ \midrule
\multirow{2}{*}{CNN} & Combination & 0.3182 & 0.2812 & 2.2553 & 2.3633 \\ \cmidrule(l){2-6} 
 & CVaR & 0.3961 & 0.4853 & 4.0703 & 4.8277 \\ \bottomrule
\end{tabular}
\label{tab:summary_equilibrium}
\end{table}

\begin{figure}[H]
    \centering
    \begin{subfigure}[b]{0.45\textwidth}
        \centering
        \includegraphics[width=\textwidth]{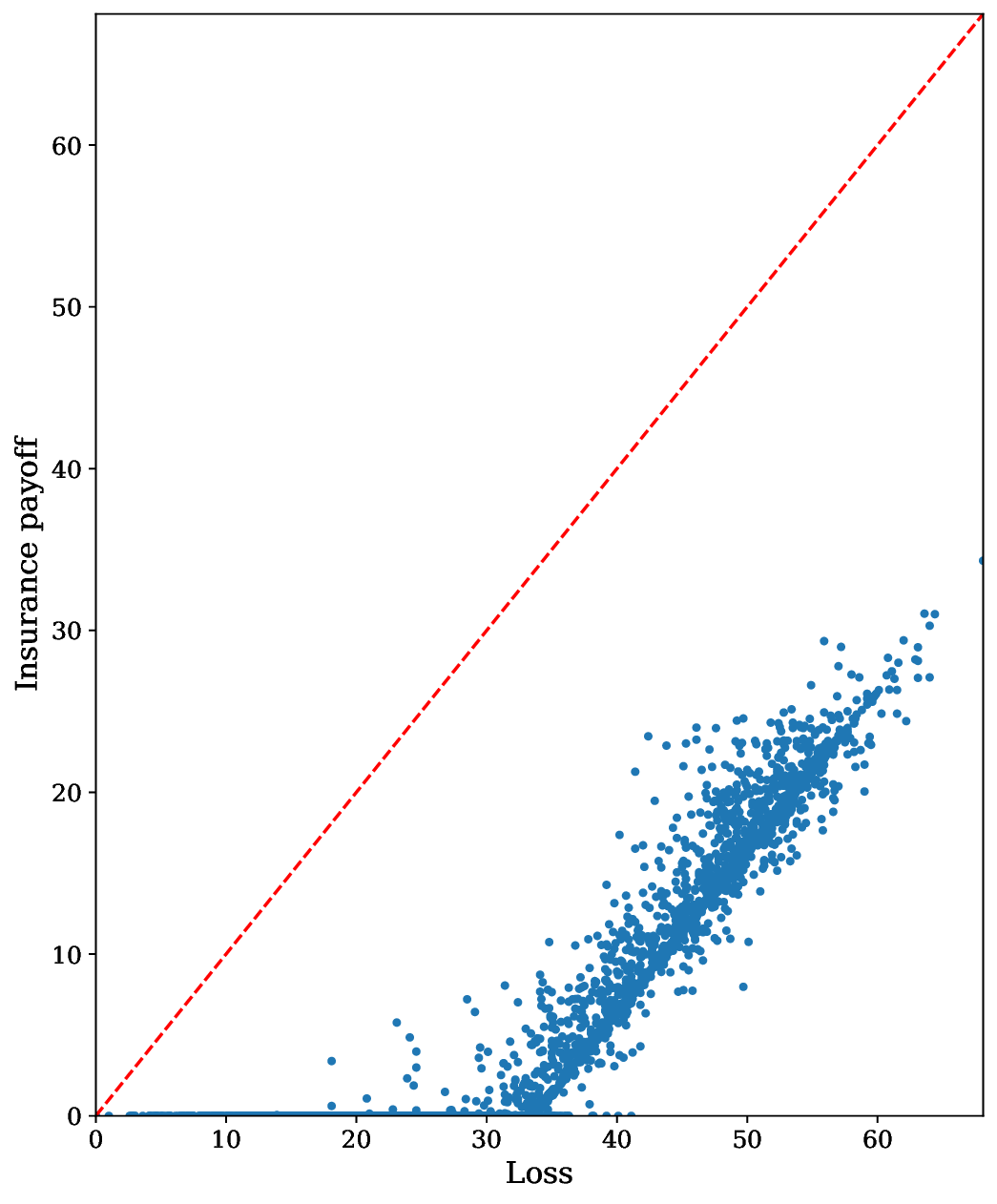}
        \caption{}
        \label{fig:index_ix_combo_cnn}
    \end{subfigure}
    \hfill
    \begin{subfigure}[b]{0.45\textwidth}
        \centering
        \includegraphics[width=\textwidth]{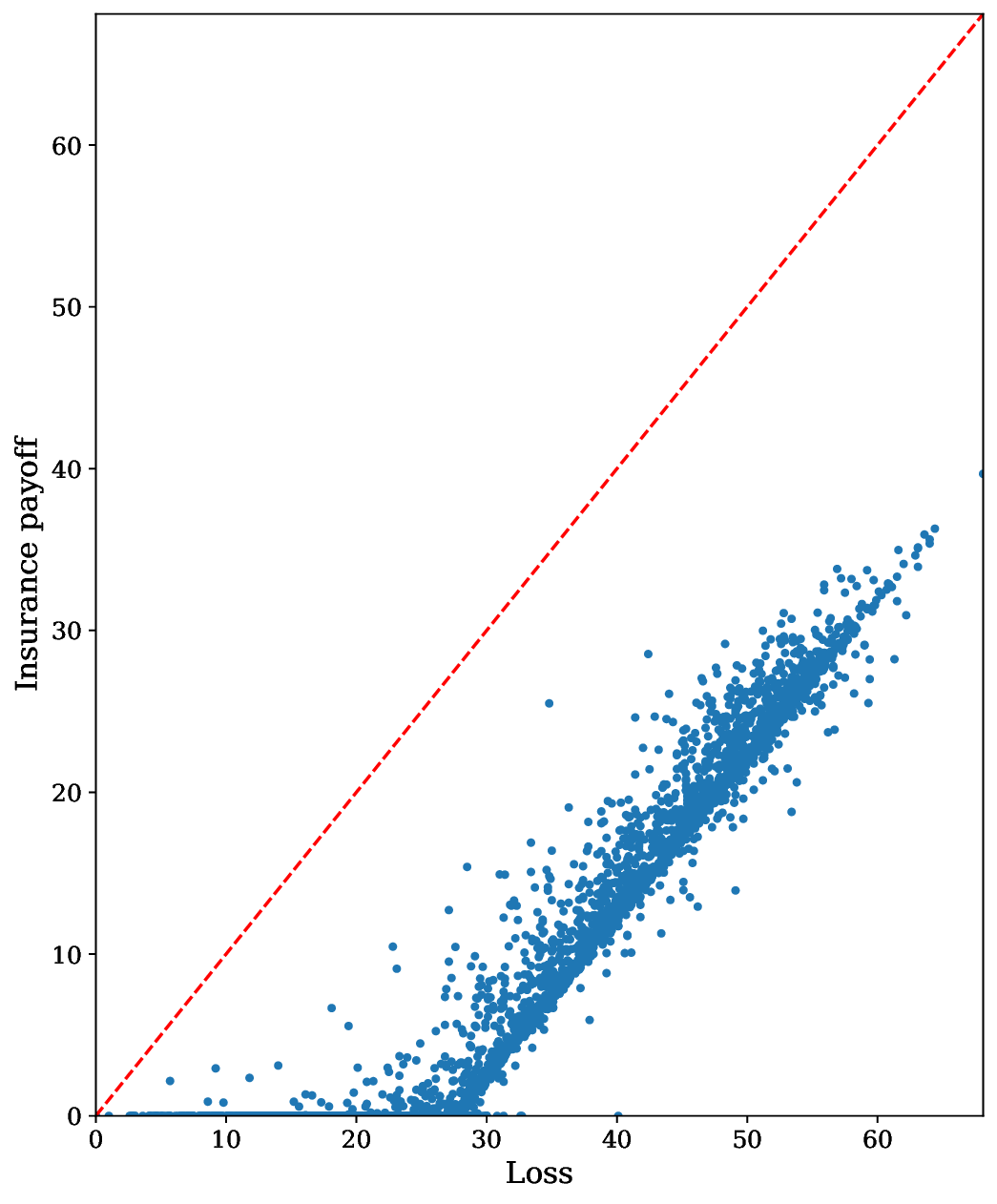}
        \caption{}
        \label{fig:index_ix_cvar_cnn}
    \end{subfigure}
    \caption{Comparison of equilibrium solutions under risk measure CVaR and convex combination with CNN used for modeling insurance payoff. Figure (a) demonstrates the optimal insurance payoff function under farmer's risk distortion $g_f(s) = 0.5s + 0.5\min\left\{\frac{s}{1-\alpha}, 1\right\}$ with $\alpha = 0.8$. The insurer's optimal risk loading is $\theta^* = 0.3182$ with a profit of 2.2553 in this case. Figure (b) shows the optimal insurance payoff function under farmer's risk measure being CVaR. 
    The insurer's optimal risk loading for this case is $\theta^* = 0.3961$ with a profit of 4.0703.}
    \label{fig:cnn_ix_cvar_combo}
\end{figure}

\subsection{Problem 2 ($\mathcal{P}_{s2}$): Two-parameter Premium Model}\label{sec:results_two_para}
This subsection presents the numerical solutions to the bilevel optimization Problem 2. In this setting, the insurer selects two parameters: risk-loading factor $\theta$ and the premium distortion factor $\rho$ in its pricing strategy. The farmer's risk measure is based on the distortion function $g_f(s) = \lambda s + (1 - \lambda) \cdot \min\{\frac{s}{1-\alpha}, 1\}, ~\text{with}~ \alpha = 0.8$. We vary the weighting factor $\lambda$ between 0.1, 0.5, and 0.7 in this problem to see differences in the equilibrium solutions. The model validation results for selecting the best model for this two-parameter optimization problem are shown in Table \ref{tab:validation_two_para}.

\begin{table}[ht]
\centering
\footnotesize
\caption{Model Validation for CNN in the two-parameter optimization case ($\lambda = 0$).}
\begin{tabular}{ccccccccc}
\hline
\multirow{2}{*}{\textbf{\begin{tabular}[c]{@{}c@{}}Convolutional \\ Layers\end{tabular}}} & \multicolumn{2}{c}{{[}16 - 16{]}} & \multicolumn{2}{c}{{[}32 - 32{]}} & \multicolumn{2}{c}{{[}16 - 16 - 16{]}} & \multicolumn{2}{c}{\textbf{{[}32 - 32- 32{]}}} \\ \cline{2-9} 
 & Training & Validation & Training & Validation & Training & Validation & \textbf{Training} & \textbf{Validation} \\ \hline
UP Loss & 0.0000 & 0.0000 & -7.7488 & -7.9577 & 0.0000 & 0.0000 & \textbf{-8.2615} & \textbf{-8.3879} \\ \hline
Penalized Loss & 0.0000 & 0.0000 & -6.3389 & -7.3672 & 0.0000 & 0.0000 & \textbf{-7.1449} & \textbf{-6.4417} \\ \hline
\multirow{2}{*}{\textbf{\begin{tabular}[c]{@{}c@{}}Fully connected layers\\ (single hidden layer)\end{tabular}}} & \multicolumn{2}{c}{{[}8{]}} & \multicolumn{2}{c}{{[}16{]}} & \multicolumn{2}{c}{{[}32{]}} & \multicolumn{2}{c}{\textbf{{[}64{]}}} \\ \cline{2-9} 
 & Training & Validation & Training & Validation & Training & Validation & \textbf{Training} & \textbf{Validation} \\ \hline
UP Loss & -7.6739 & -7.7602 & -7.5128 & -7.5883 & -7.7870 & -8.1214 & \textbf{-8.2000} & \textbf{-8.3904} \\ \hline
Penalized Loss & -4.6893 & -8.5636 & -9.2603 & -10.5043 & -8.6762 & -10.6711 & \textbf{-7.6511} & \textbf{-9.4009} \\ \hline
\multirow{2}{*}{\textbf{\begin{tabular}[c]{@{}c@{}}Fully connected layers\\ (multiple hidden layers)\end{tabular}}} & \multicolumn{2}{c}{{[}64 - 64{]}} & \multicolumn{2}{c}{{[}64 - 64 - 64{]}} & \multicolumn{2}{c}{{[}64 - 64 - 64 - 64{]}} & \multicolumn{2}{c}{{[}64 - 64 - 64 - 64 - 64{]}} \\ \cline{2-9} 
 & Training & Validation & Training & Validation & Training & Validation & Training & Validation \\ \hline
UP Loss & -8.3950 & -8.5892 & -8.2634 & -8.5726 & -8.0183 & -8.0700 & 0.0000 & 0.0000 \\
Penalized Loss & -16.4543 & -12.8822 & -12.2764 & -5.6959 & -4.2467 & -7.8628 & 0.0000 & 0.0000 \\ \hline
\end{tabular}
\label{tab:validation_two_para}
\begin{minipage}{\linewidth}
\raggedright\footnotesize Note: The optimal model is the one with the lowest validation UP loss given that the absolute difference between validation UP loss and validation penalized loss is below 2.
\end{minipage}
\end{table}

Figure \ref{fig:two_para_index_indemnity} shows the optimal insurance payoff functions for different values of $\lambda$ in the farmer's risk distortion function $g_f(s) = \lambda s + (1-\lambda)\min\left\{\frac{s}{1-\alpha}, 1\right\}$. For $\lambda = 0.1$, the optimal values for the two parameters are $\theta^* = 0.2237$ and $\rho = 2.0644$, which corresponds to a profit of 7.4848. 
For $\lambda = 0.5$, the equilibrium solutions are $\theta^* = 0.0758$ and $\rho = 1.7341$, which leads to an expected profit of
4.2511. For $\lambda = 0.7$, the equilibrium solutions are $\theta^* = 0.0190$ and $\rho = 1.5016$, which corresponds to a profit of 2.1893. The insurance payoff pattern exhibits the limited stop-loss insurance with both a deductible and a maximum benefit. A more risk-averse farmer ($\lambda = 0.1$) leads to larger $\theta$ and $\rho$ compared to a less risk-averse farmer ($\lambda = 0.7$) as shown in Figure \ref{fig:lambda01_index_two_para} and Figure \ref{fig:lambda07_index_two_para}. For an even less risk-averse farmer ($\lambda > 0.7)$, there is no insurance.

We compare the solutions for index insurance with those for indemnity insurance. The profits obtained for indemnity insurance are higher for all three values of $\lambda$. A summary of the equilibrium solutions for Problem 2 is shown in Table \ref{tab:summary_two_para}.

\begin{table}[htbp]
\caption{Summary of Equilibrium Solutions for Bilevel Optimization Problem 2.}
\centering
\begin{tabular}{@{}cccllcc@{}}
\toprule
\multirow{2}{*}{\textbf{$\lambda$}} & \multicolumn{2}{c}{\textbf{$\theta^*$}} & \multicolumn{2}{c}{\textbf{$\rho^*$}} & \multicolumn{2}{c}{\textbf{Expected Profit}} \\ \cmidrule(l){2-7} 
 & \multicolumn{1}{l}{Index} & \multicolumn{1}{l}{Indemnity} & Index & Indemnity & \multicolumn{1}{l}{Index} & \multicolumn{1}{l}{Indemnity} \\ \midrule
0.1 & 0.2237 & 0.2761 & 2.0644 & 1.7078 & 7.4848 & 11.3735 \\ \midrule
0.5 & 0.0758 & 0.1022 & 1.7341 & 1.5641 & 4.2511 & 6.7666 \\ \midrule
0.7 & 0.0190 & 0.0258 & 1.5016 & 1.4565 & 2.1893 & 4.3443 \\ \bottomrule
\end{tabular}
\label{tab:summary_two_para}
\end{table}

\begin{figure}[h!]
    \centering
    \begin{subfigure}[b]{0.3\textwidth} 
        \includegraphics[width=\textwidth]{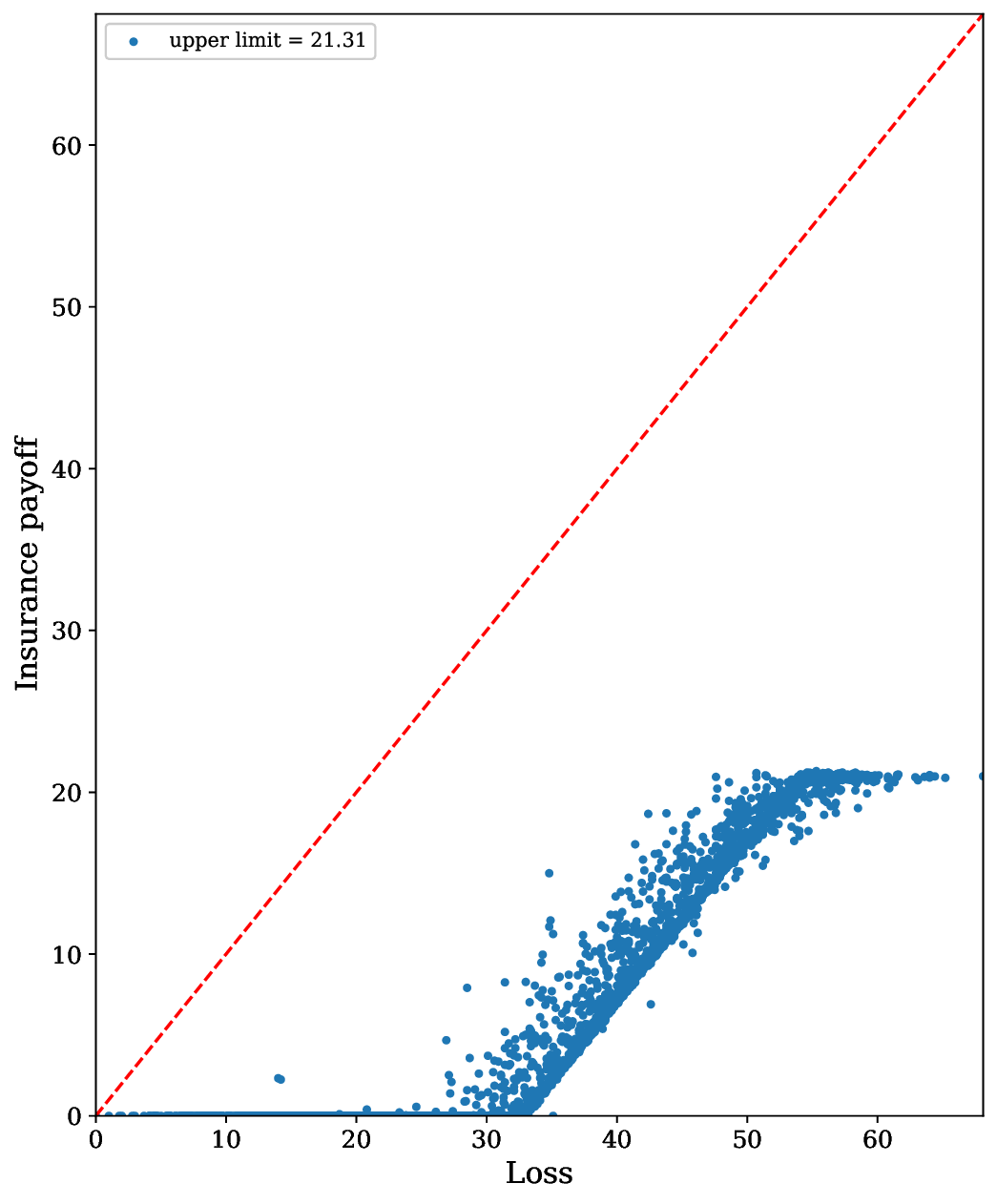}
        \caption{ES:$\theta = 0.2237$ and $\rho = 2.0644$}
        \label{fig:lambda01_index_two_para}
    \end{subfigure}
    \hfill
    \begin{subfigure}[b]{0.3\textwidth} 
        \includegraphics[width=\textwidth]{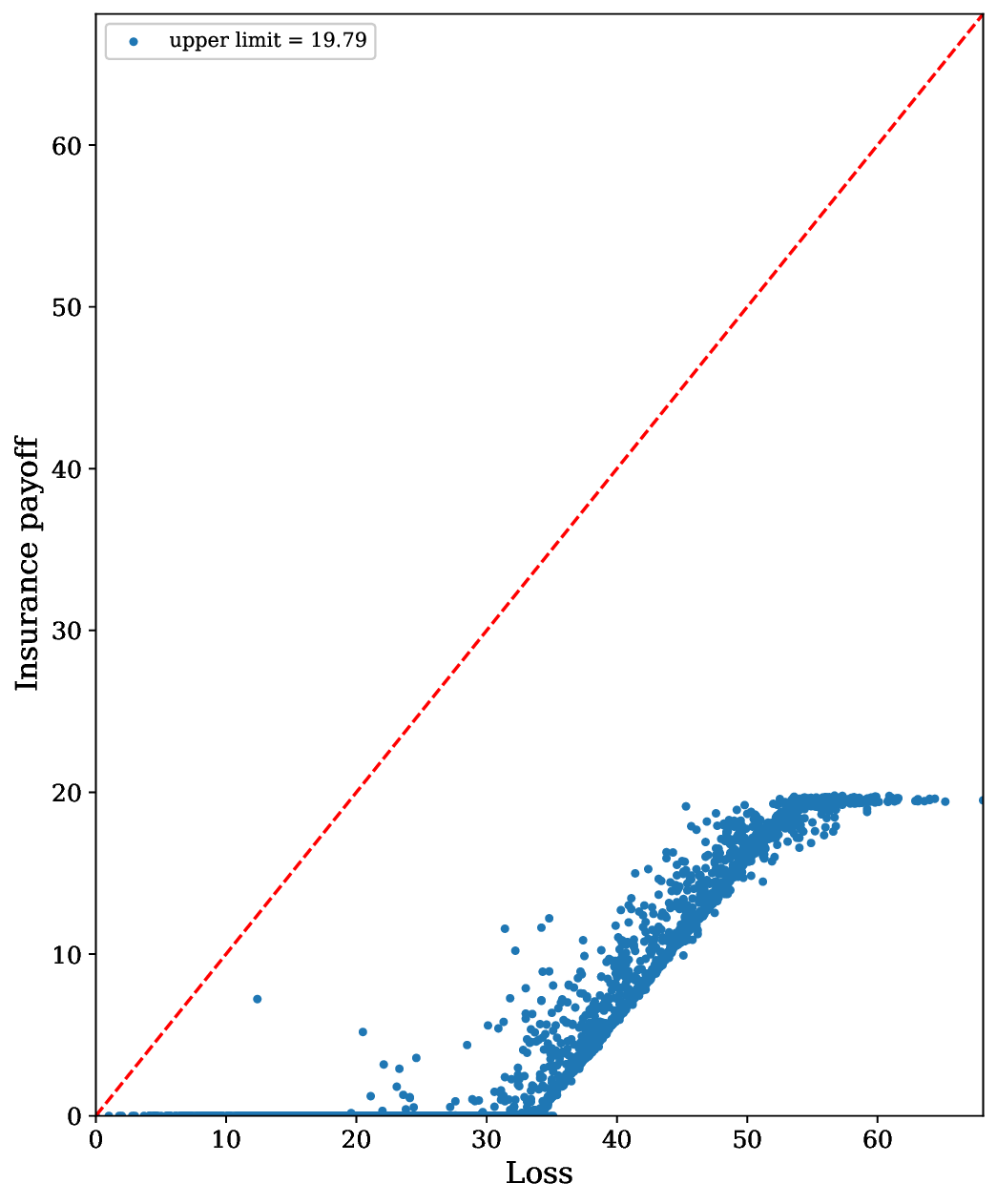}
        \caption{ES:$\theta = 0.0758$ and $\rho = 1.7341$}
        \label{fig:lambda05_index_two_para}
    \end{subfigure}
    \hfill
    \begin{subfigure}[b]{0.3\textwidth} 
        \includegraphics[width=\textwidth]{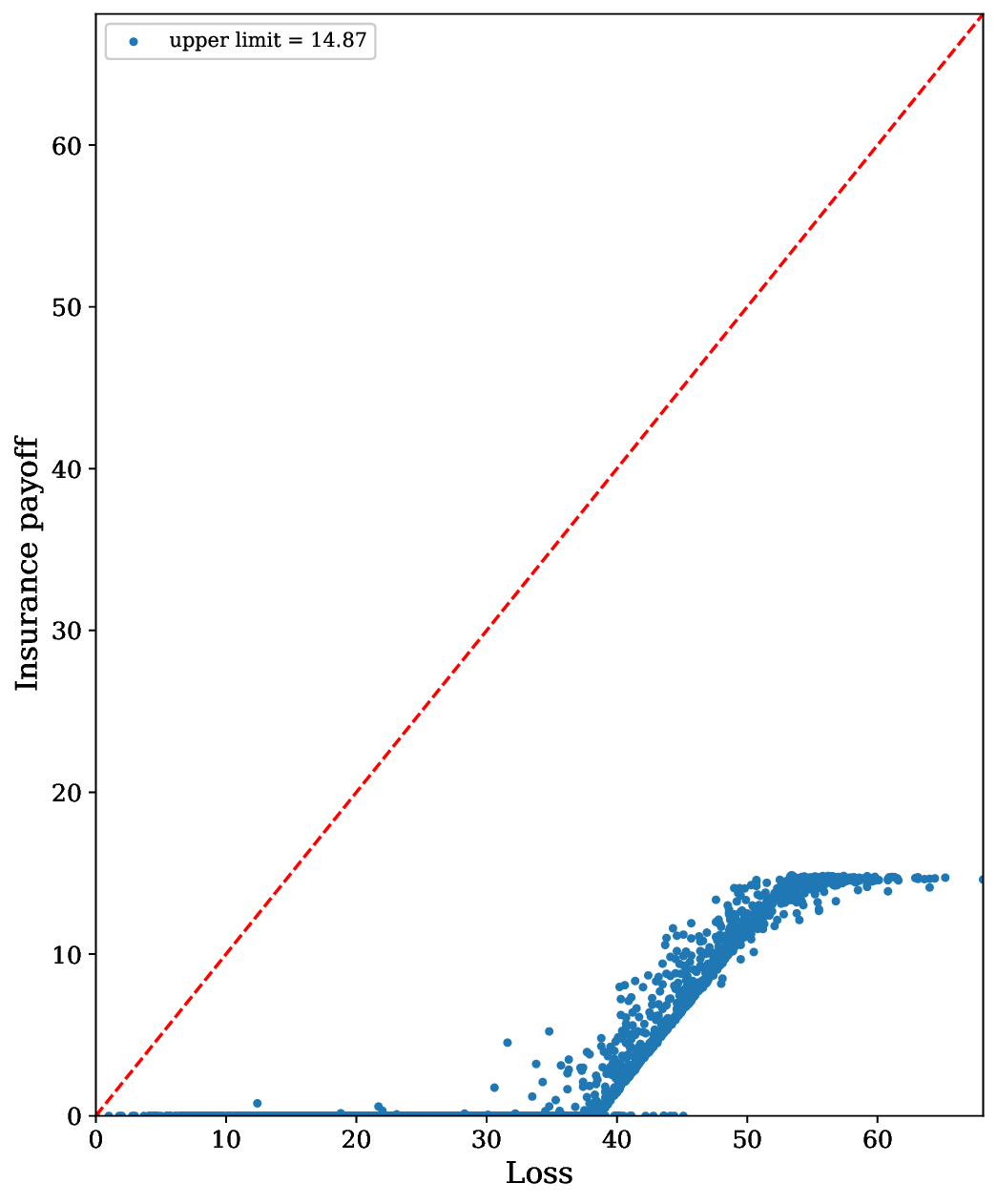}
        \caption{ES: $\theta = 0.0190$ and $\rho = 1.5016$}
        \label{fig:lambda07_index_two_para}
    \end{subfigure}
    
    \vspace{10pt}
    
    \begin{subfigure}[b]{0.3\textwidth} 
        \includegraphics[width=\textwidth]{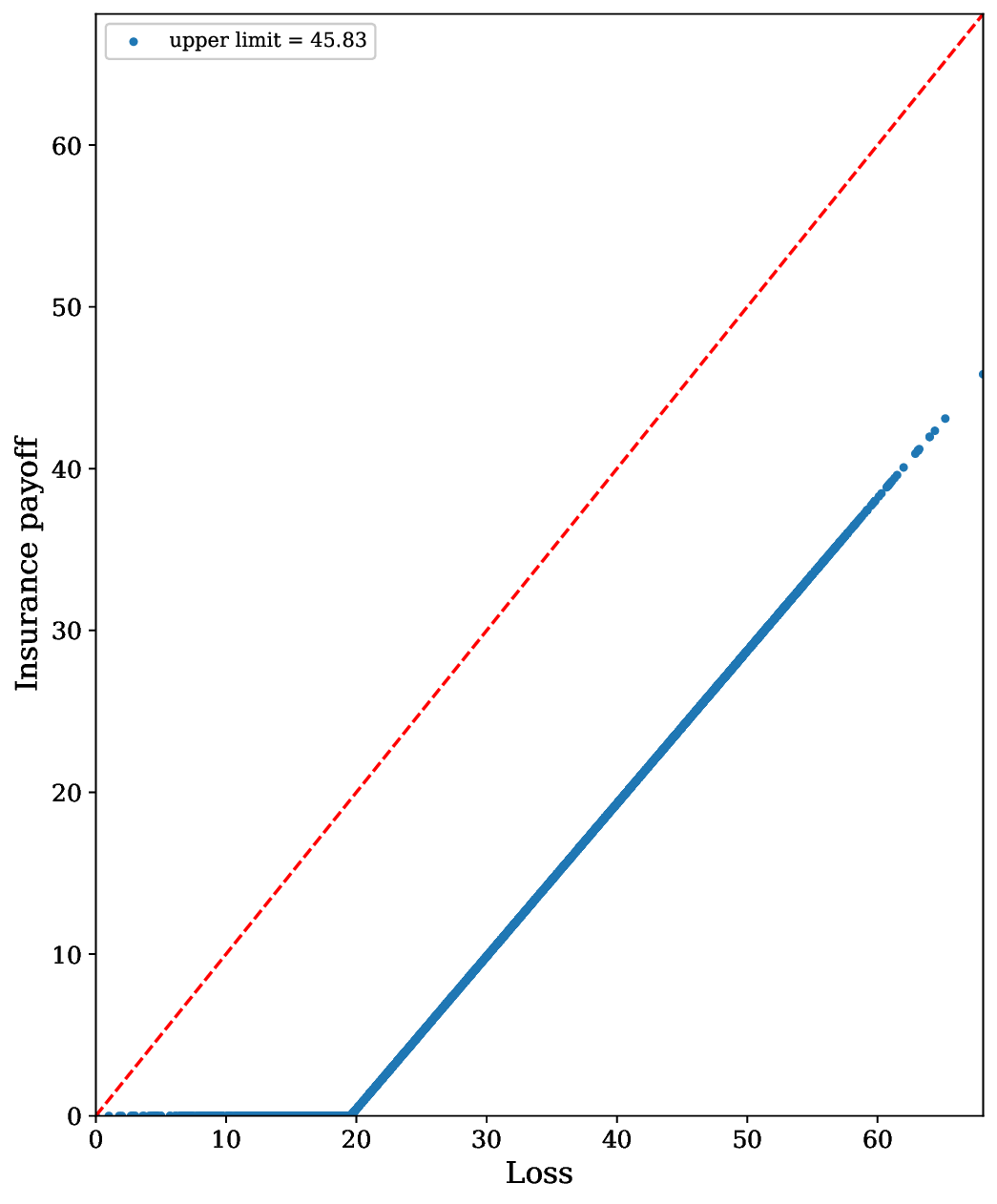}
        \caption{ES: $\theta = 0.2761$ and $\rho = 1.7078$}
        \label{fig:lambda01_indem_two_para}
    \end{subfigure}
    \hfill
    \begin{subfigure}[b]{0.3\textwidth} 
        \includegraphics[width=\textwidth]{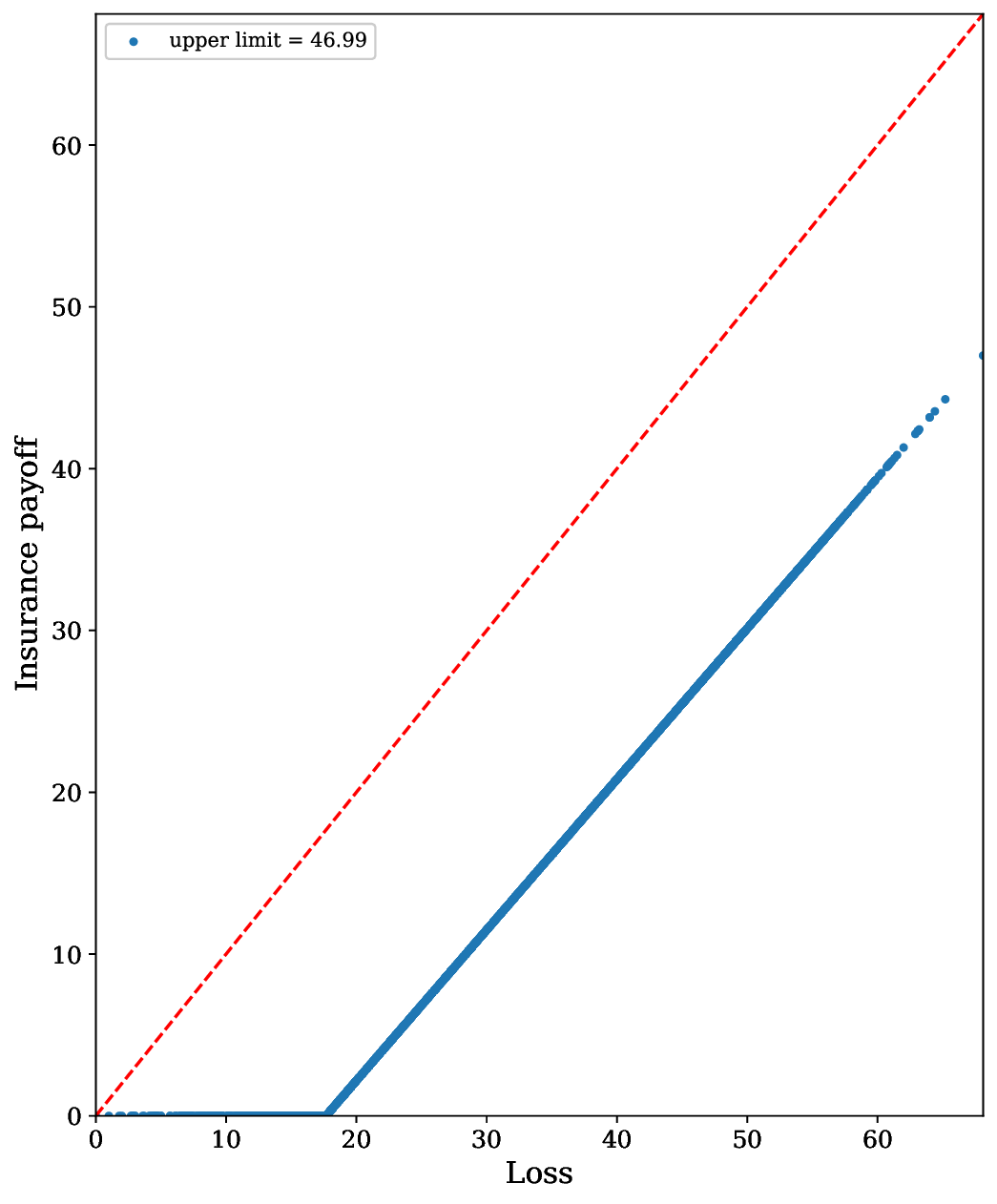}
        \caption{ES: $\theta = 0.1022$ and $\rho = 1.5641$}
        \label{fig:lambda05_indem_two_para}
    \end{subfigure}
    \hfill
    \begin{subfigure}[b]{0.3\textwidth} 
        \includegraphics[width=\textwidth]{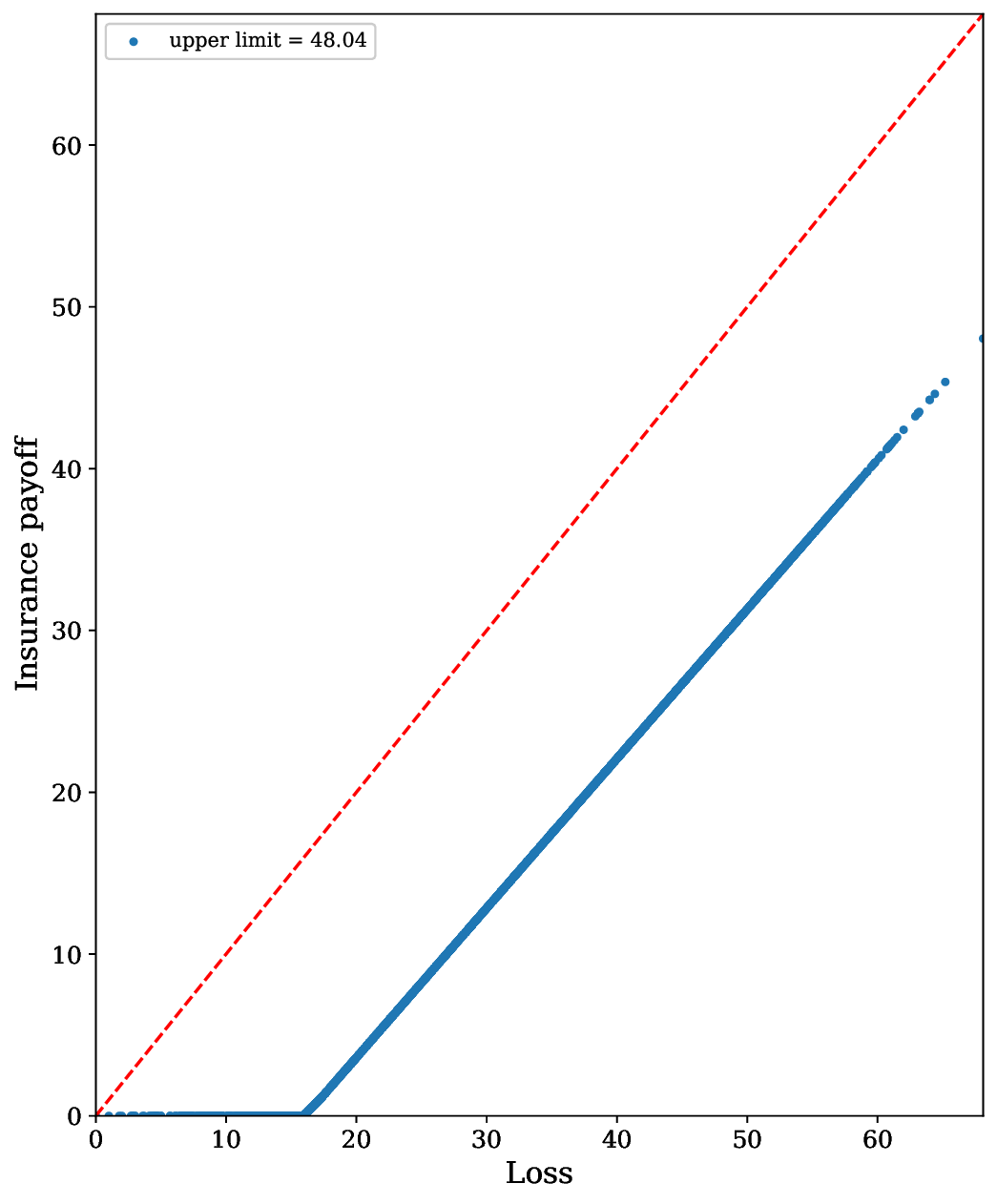}
        \caption{ES: $\theta = 0.0258$ and $\rho = 1.4565$}
        \label{fig:lambda07_indem_two_para}
    \end{subfigure}
    
    \caption{Equilibrium solutions under the farmer's risk distortion function $g_f(s) = \lambda s + (1-\lambda)\min\left\{\frac{s}{1-\alpha}, 1\right\}, \alpha = 0.8$ with different values of $\lambda$ and insurer's profit maximization using two parameters. Figure (a), (b) and (c) show the index insurance payoff functions for $\lambda = 0.1, 0.5~\text{and} ~0.7$ respectively in farmer's risk distortion function. Figures (d), (e) and (f) show the indemnity insurance payoff functions for $\lambda = 0.1, 0.5~\text{and} ~0.7$ respectively in farmer's risk distortion function.}
    \label{fig:two_para_index_indemnity}
\end{figure}

\subsection{Problem 3 ($\mathcal{P}_{s}$): General Premium Model}\label{sec:results_pk}

This subsection presents the solutions to Problem 3. In this case, we do not impose any structural form of the premium function and instead, we search for the optimal pricing function $g_i(s)$ for the premium $\hat{\Pi}(I(\mathbf{X}))=\int_0^{\infty} g_i(\mathbb{P}(I(\mathbf{X})>z)) dz$. 
\begin{figure}[htbp]
    \centering
    \begin{subfigure}[b]{0.45\textwidth}
        \centering
        \includegraphics[width=\textwidth]{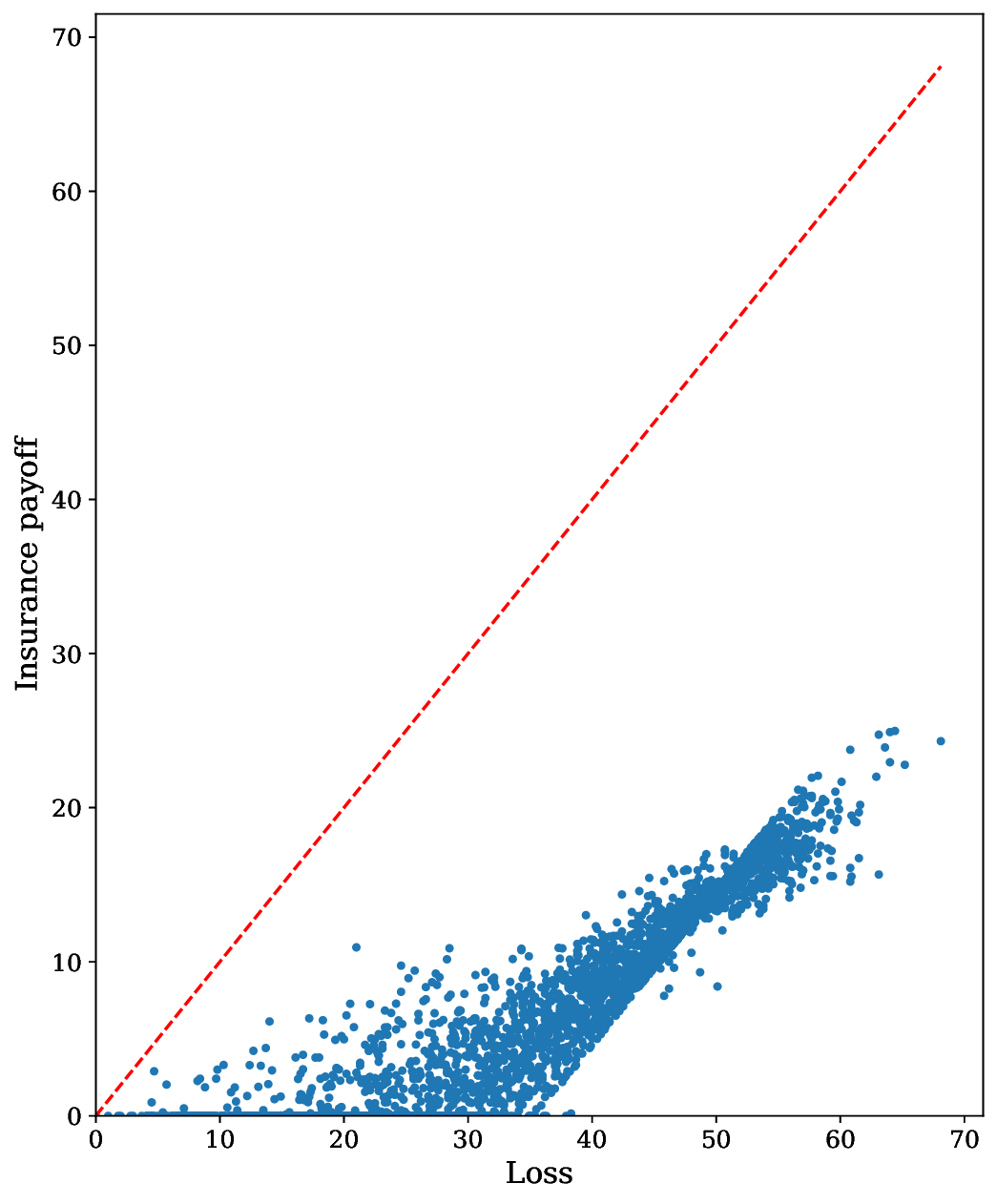}
        \caption{}
        \label{fig:ix_pk_cnn}
    \end{subfigure}
    \hfill
    \begin{subfigure}[b]{0.45\textwidth}
        \centering
        \includegraphics[width=\textwidth]{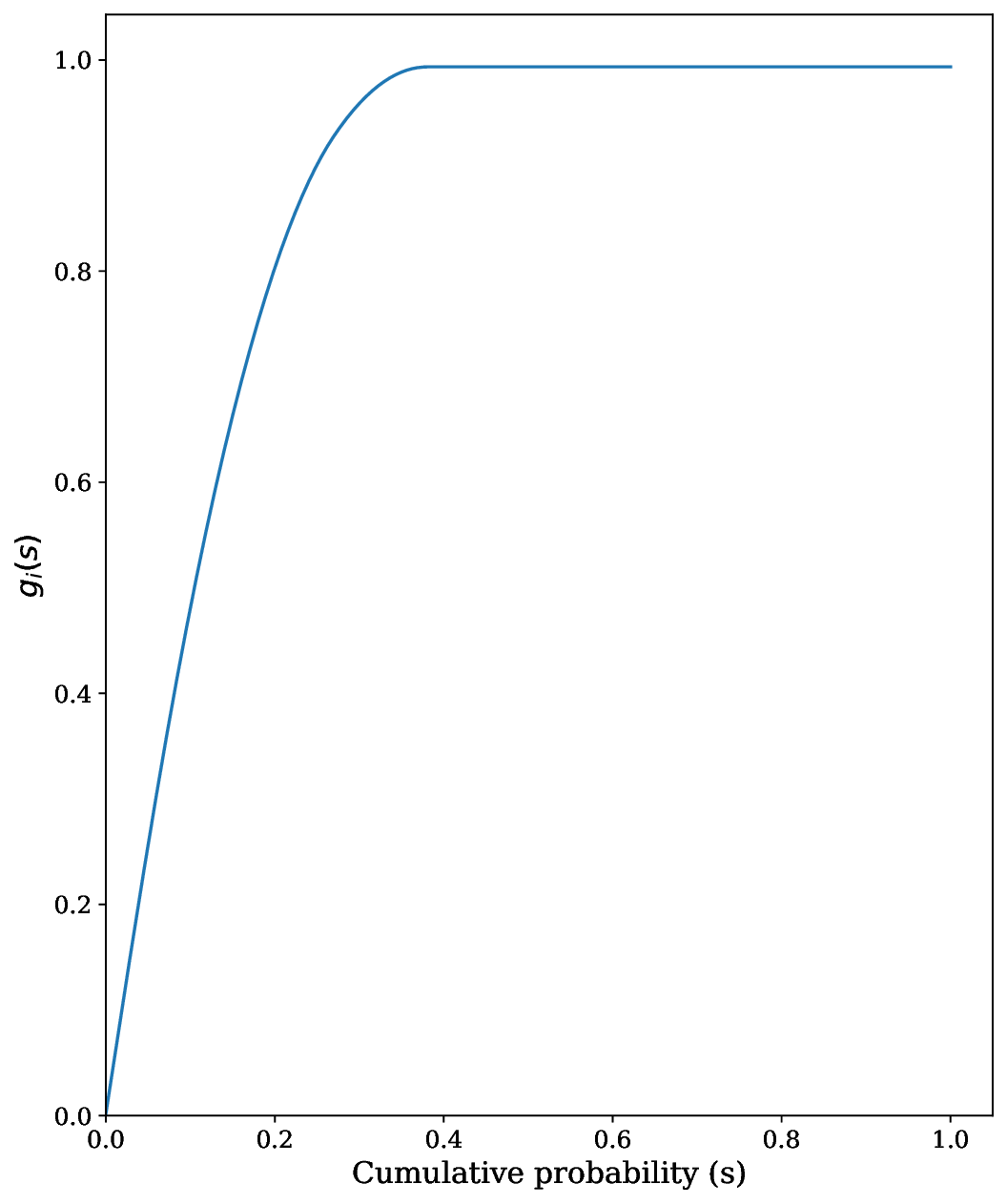} 
        \caption{}
        \label{fig:gx_pk_cnn}
    \end{subfigure}
    \vspace{10pt}
    \begin{subfigure}[b]{0.45\textwidth}
        \includegraphics[width=\textwidth]{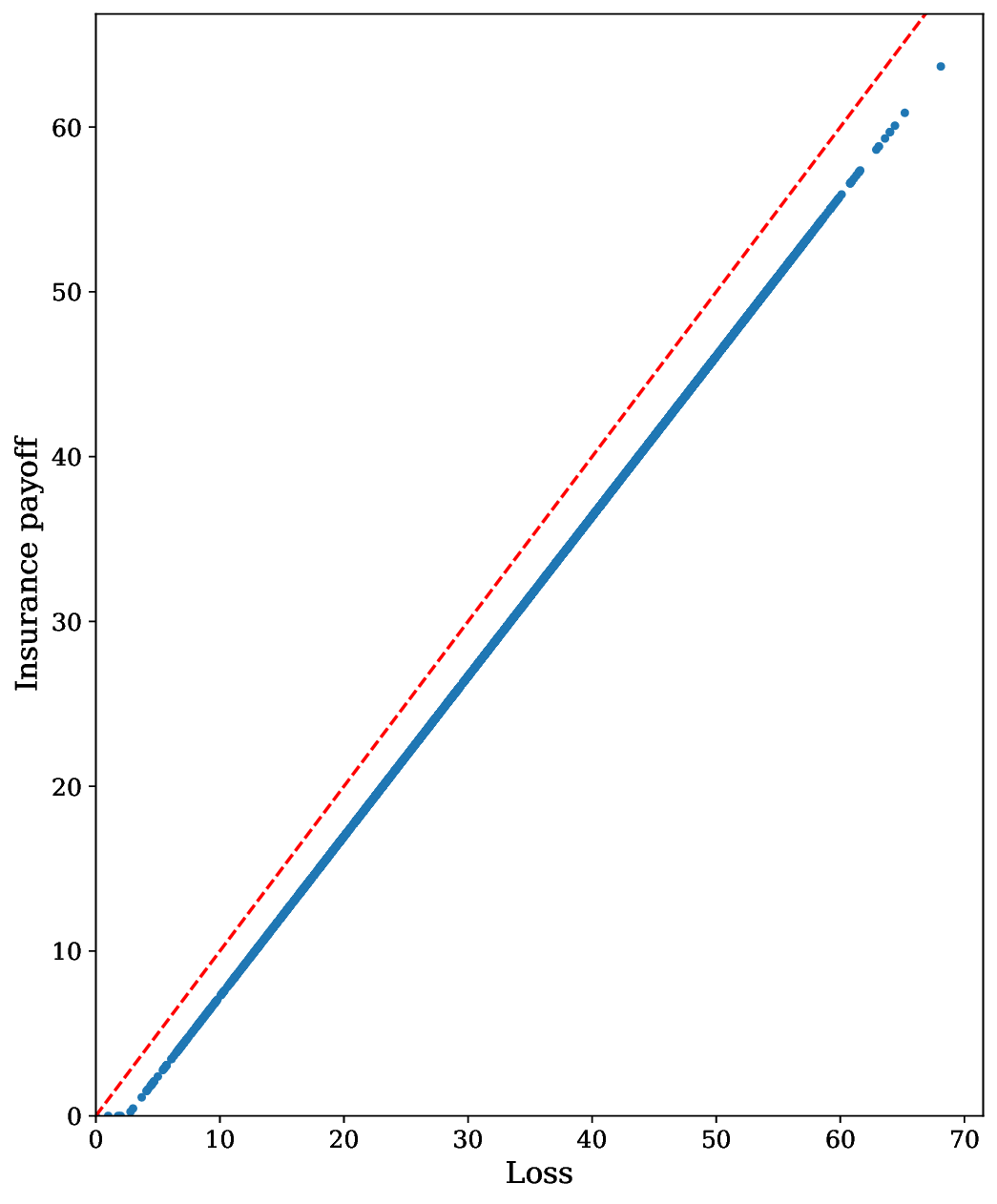}
        \caption{}
        \label{fig:ix_pk_indem}
    \end{subfigure}
    \hfill
    \begin{subfigure}[b]{0.45\textwidth}
        \includegraphics[width=\textwidth]{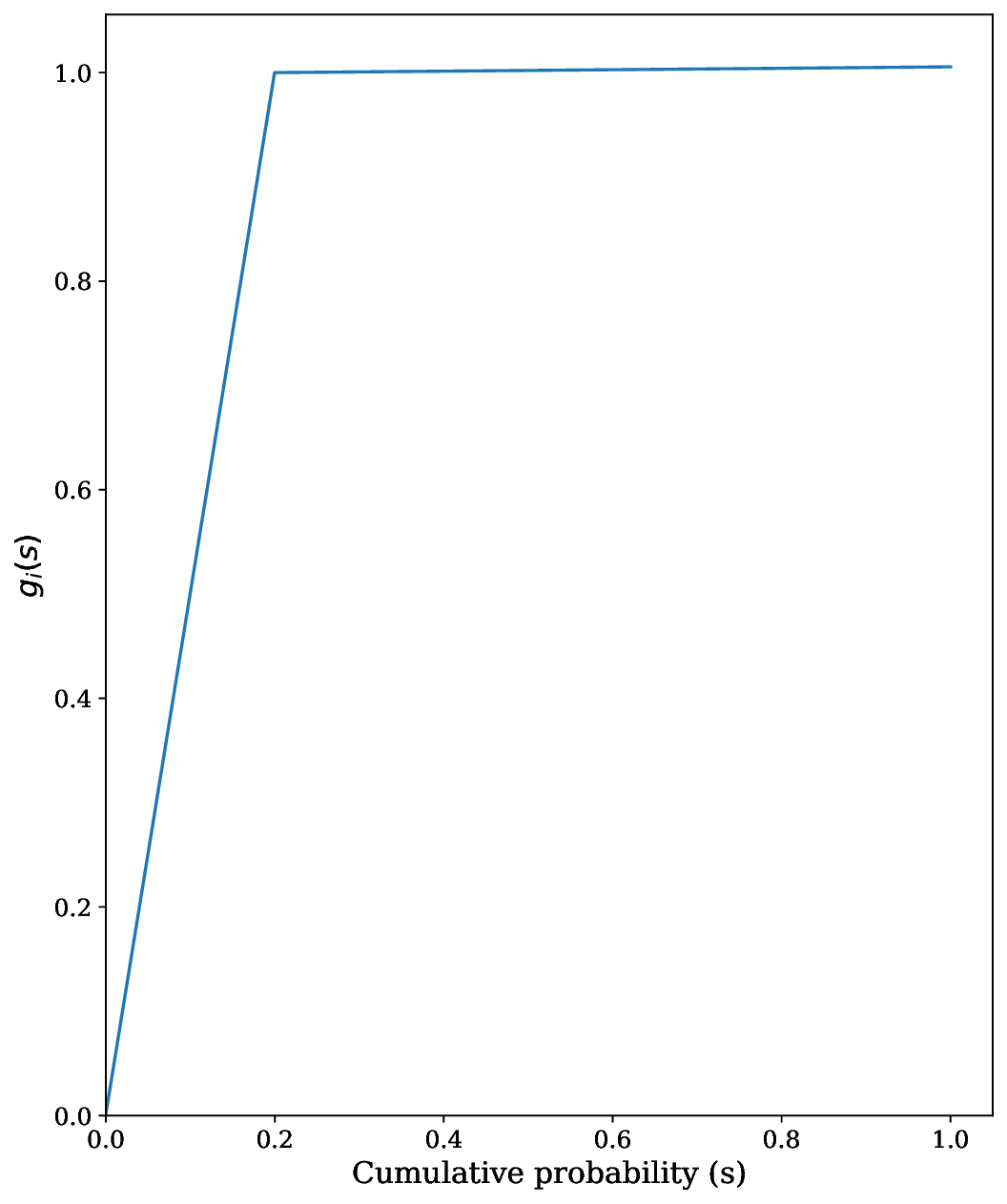}
        \caption{}
        \label{fig:gx_pk_indem}
    \end{subfigure}
    \caption{Equilibrium solution under Insurer's General Pricing Strategy. Figure (a) shows the optimal insurance payoff function when the insurer adopts a general pricing strategy through a premium distortion function $g_i(s)$. The farmer's risk distortion function adopted here is $g_f(s) = 0.5s + 0.5\min\left\{\frac{s}{1-\alpha}, 1\right\} ~\text{with}~ \alpha = 0.8$. Figure (b) shows the corresponding optimal pricing function $g_i$ of the insurer. Figures (c) and (d) show the optimal insurance payoff functions and pricing function obtained for indemnity insurance.}
    \label{fig:pricing_kernel_ix_gx}
\end{figure}

Figure \ref{fig:pricing_kernel_ix_gx} shows the insurance payoff function and the pricing function $g_i(s)$ obtained in equilibrium. For this case, we search for the optimal pricing function by using the optimal insurance payoff function obtained via CNNs in Section \ref{sec:results_cnn} as the initial values. The pricing function is modeled using NNs and the network architecture is selected from Table \ref{tab:table_pk}. 
\begin{table}[ht]
\centering
\footnotesize
\caption{Model validation for the general pricing function (Problem 3).}
\begin{tabular}{ccccccccc}
\hline
\multirow{2}{*}{\textbf{Single hidden layer}} & \multicolumn{2}{c}{{[}8{]}} & \multicolumn{2}{c}{\textbf{{[}16{]}}} & \multicolumn{2}{c}{{[}32{]}} & \multicolumn{2}{c}{{[}64{]}} \\ \cline{2-9} 
 & Train & Validation & \textbf{Train} & \textbf{Validation} & Train & Validation & Train & Validation \\ \hline
UP Loss & -7.4560 & -7.6460 & \textbf{-4.2022} & \textbf{-4.1166} & 0.0000 & 0.0000 & 0.0000 & 0.0000 \\ \hline
Penalized Loss & -4.3509 & -59.5728 & \textbf{-2.7396} & \textbf{-3.1015} & 0.000 & 0.0000 & 0.0000 & 0.0000 \\ \hline
\multirow{2}{*}{\textbf{Multiple hidden layers}} & \multicolumn{2}{c}{{[}16 - 16{]}} & \multicolumn{2}{c}{{[}16 - 16 - 16{]}} & \multicolumn{2}{c}{\textbf{{[}16 - 16 - 16 - 16{]}}} & \multicolumn{2}{c}{{[}16 - 16 - 16 - 16 - 16{]}} \\ \cline{2-9} 
 & Train & Validation & Train & Validation & \textbf{Train} & \textbf{Validation} & Train & Validation \\ \hline
UP Loss & -9.2257 & -9.4739 & -9.3566 & -9.5601 & \textbf{-8.7961} & \textbf{-8.9043} & -9.0715 & -9.2617 \\ \hline
Penalized Loss & -9.1168 & -14.3956 & -9.1887 & -27.2854 & \textbf{-8.6862} & \textbf{-6.9138} & -8.9940 & -40.7818 \\ \hline
\end{tabular}
\label{tab:table_pk}
\begin{minipage}{\linewidth}
\raggedright\footnotesize Note: The optimal model is the one with the lowest validation UP loss given that the absolute difference between validation UP loss and validation penalized loss is below 2.
\end{minipage}
\end{table}

Figure \ref{fig:ix_pk_cnn} shows the optimal index insurance payoff function, and Figure \ref{fig:gx_pk_cnn} shows the corresponding optimal pricing function. In both figures, the farmer’s risk measure is CVaR. The profit obtained is 8.7961, which is higher than the profits obtained in Problems 1 and 2 under the same risk measure for the farmer (CVaR). The monopolistic insurer has much more freedom in this case and is thus capable of extracting more profit. The optimal solutions under the same farmer's risk measure for indemnity-based insurance are shown in Figures \ref{fig:ix_pk_indem} and \ref{fig:gx_pk_indem}. The optimal insurance payoff function shown in Figure \ref{fig:ix_pk_indem} has a deductible close to zero. The optimal pricing kernel shown in Figure \ref{fig:gx_pk_indem} is piecewise linear with a kink at $s = 0.2$. The profit obtained in the optimal solution for the indemnity-based insurance is 17.2982. It is interesting to note that the shape of the optimal pricing kernel for the indemnity-based insurance is piecewise linear, while that for the index insurance is smooth and concave. We can verify the results for indemnity-based insurance in this subsection using the techniques in \cite{cheung2019risk}.

\section{Conclusion}\label{sec:6}
This study provides a comprehensive analysis of monopoly pricing in weather-index-based insurance. We formulate the pricing problem as a sequential game between a profit-maximizing insurer and a risk-averse farmer and solve it using a penalized bilevel programming algorithm. Moreover, NNs and CNNs are employed to determine the optimal insurance strategy, and it is found that stop-loss insurance emerges as the preferred contract structure, particularly under risk measures such as CVaR. We contribute to the literature by introducing CNNs into index insurance pricing, yielding more robust insurance payoff functions than fully connected NNs. This implies that CNNs can reduce the basis risk in index insurance pricing, enabling the insurer to achieve profits closer to those of an indemnity insurance contract.

Our analysis further reveals the critical role of the farmer's risk appetite in shaping the market equilibrium. While more risk-averse farmers demand lower deductibles, the monopolistic insurer strategically responds with more aggressive pricing. We show this by comparing equilibrium solutions across three cases with varying degrees of pricing flexibility. When the insurer can optimize either a two-parameter premium (a risk-loading factor $\theta$ and a premium distortion parameter $\rho$) or a general pricing function, it can extract significantly higher profits from more risk-averse farmers. 

This paper introduces a deep bilevel programming framework for pricing weather index insurance. It highlights how advanced machine learning and game theory can be applied to effectively analyze the strategic behavior in agricultural markets. Future research could extend this framework to examine how the different market structures, such as duopoly or perfect competition, shape the contract design and farmers' welfare. Such extensions would provide policymakers and industry stakeholders with important information to manage agricultural risk efficiently.

\bibliographystyle{apalike}     
\bibliography{reference} 
\end{document}